\RequirePackage{amsmath}

\documentclass[smallextended]{svjour3}       
\smartqed 

\usepackage[utf8]{inputenc}
\usepackage[T1]{fontenc}
\usepackage{graphicx}
\usepackage{txfonts}
\usepackage{booktabs,multirow,array}
\usepackage{colortbl}
\usepackage[vlined,algoruled]{algorithm2e}
\usepackage{algorithmic}
\usepackage[english]{babel}
\usepackage{thm-restate}
\usepackage{tikz}
\usepackage{enumerate}
\usepackage{wrapfig}
\usetikzlibrary{decorations.pathreplacing,calc,positioning,decorations.markings,patterns}
\usepackage {nicefrac}

\usepackage{subcaption}
\usepackage{multirow}
\allowdisplaybreaks

\usepackage{url}
\usepackage{xcolor}
\usepackage{transparent}
\usepackage{txfonts}
\usepackage{todonotes}
\usepackage{float}

\definecolor{darkgreen}{RGB}{63,167,129}
\definecolor{darkblue}{RGB}{41,60,134}
\definecolor{darkorange}{RGB}{255,167,15}
\definecolor{lightgray}{RGB}{127,127,127}

\newcommand{\RR}{\mathbb{R}}
\newcommand{\ZZ}{\mathbb{Z}}
\newcommand{\NN}{\mathbb{N}}
\newcommand{\scenarios}{\mathcal{S}}
\newcommand{\scenariosDualPath}[2]{\mathcal{S}(#1,#2)}
\newcommand{\paths}{\mathcal{P}}

\newcommand{\sectionheadline}[1]{\paragraph{\tb{#1.}}}

\newcommand{\tb}[1]{\textbf{#1}}
\newcommand{\ti}[1]{\textit{#1}}

\newcommand{\omsClaim}[1]{\vspace{.2cm}\parbox{.95\linewidth}{\ti{Claim:}\; #1}\vspace{.2cm} \\}

\title{Robust Flows over Time: Models and Complexity Results}
\author{Corinna Gottschalk$^1$ \and Arie M.C.A. Koster$^2$ \and Frauke Liers$^3$ \and Britta Peis$^1$ \and Daniel Schmand$^1$ \and Andreas Wierz$^1$}
\institute{$^1$School of Business and Economics, RWTH Aachen University, Germany, \email{\{gottschalk, peis, schmand, wierz\}@oms.rwth-aachen.de}
\and $^2$Lehrstuhl II für Mathematik, RWTH Aachen University, Germany, \email{koster@math2.rwth-aachen.de}
\and $^3$Lehrstuhl Wirtschaftsmathematik, FAU Erlangen-Nürnberg, Germany, \email{frauke.liers@math.uni-erlangen.de}}

\begin{document}

\maketitle

\begin{abstract}
We study dynamic network flows with uncertain input data under a robust optimization perspective.
In the dynamic maximum flow problem, the goal is to maximize the flow reaching the sink within a given time horizon $T$, while flow requires a certain travel time to traverse an edge.

In our setting, we account for uncertain travel times of flow.
We investigate maximum flows over time under the assumption that at most $\Gamma$ travel times may be prolonged simultaneously due to delay.
We develop and study a mathematical model for this problem.
As the dynamic robust flow problem generalizes the static version, it is NP-hard to compute an optimal flow.
However, our dynamic version is considerably more complex than the static version. We show that it is NP-hard to verify feasibility of a given candidate solution.
Furthermore, we investigate temporally repeated flows and show that in contrast to the non-robust case (that is, without uncertainties) they no longer provide optimal solutions for the robust problem, but rather yield a worst case optimality gap of at least $T$.
We finally show that the optimality gap is at most $O(\eta k \log T)$, where $\eta$ and $k$ are newly introduced instance characteristics and provide a matching lower bound instance with optimality gap $\Omega(\log T)$ and $\eta = k = 1$.
The results obtained in this paper yield a first step towards understanding robust dynamic flow problems with uncertain travel times.
\end{abstract}

\section{Introduction}
\label{sec:introduction}

Many relevant applications in the context of routing or logistics call for a temporal component that is part of the input and the actual solution.
In classical flow theory, flow traverses the network in a static fashion, that is, the solutions need to obey capacity restrictions - and possibly additional constraints.
In many real-world applications, however, flow takes some time in order to traverse a network edge.
Hence, a temporal dimension has to be introduced into the models.
Dynamic flow problems that take time into account have been studied for more than half a century.
Dynamic flow problems are also referred to as \emph{flow over time} problems in the literature and in this paper.
Despite the relevance of the topic and the existence of fascinating results, few textbooks cover this topic.
We refer to \cite{skutella2009introduction} for an introduction to flow over time problems.
Furthermore, real-world applications are often also affected by measurement errors as well as by a large degree of uncertainty.
In such situations, the classical models are usually highly inaccurate.
An application may be distribution networks such as gas
\cite{koch2015evaluating} or water networks.
In gas networks, for example, the roughness of the pipes is uncertain due to contamination or aging processes and can only be measured with large effort.
The roughness strongly influences the friction and thus the travel time of gas along a pipe.
As these uncertainties might influence the decision about feasibility or infeasibility of the corresponding complex optimization tasks, a worst-case robust perspective is appropriate.

Although relevant applications exist, little is known about flow over time problems that are protected against uncertain conditions.
In this work, we provide a first step towards studying their properties.
We first define an appropriate model for determining robust optimum flows over time subject to uncertain travel times.
We start from the known classical flow over time theory that ignores uncertainties, that is, the nominal case. Adapting the $\Gamma$-robustness model introduced by
Bertsimas and Sim \cite{BeSi03} that is often applied to combinatorial
optimization problems under uncertainty, we develop a robust flow over time framework.
In this model, the degree of protection against uncertainties can be controlled by a parameter $\Gamma$.
For uncertain objective functions, a robust optimum is a solution with the best guaranteed cost under the assumption that at most $\Gamma$ many objective function coefficients attain their worst-case realization.
Thus, the value of $\Gamma$ determines the conservatism of the solution.

Applying this modeling framework for flows over time, we study the following optimization task.
Let a time horizon $T$ and a network with travel times and potential delays on the edges be given.
Maximize the minimum amount of flow that can be sent through the network within time horizon $T$ under the assumption that at most $\Gamma$ edges are delayed.
In order to study the problem in its most basic version, we restrict ourselves to the case in which flow is not allowed to wait at any intermediate vertex.
For example, this is the case in communication, water, or gas networks without buffer capacities  on intermediate vertices.
Although the obtained model seems quite punishing, it is a natural robust counterpart of network flows over time with uncertain travel times.
We provide a formal problem definition in Section~\ref{sec:modelling-techniques}.

\sectionheadline{Related Work}
\label{sec:related-work}
The concept of flows over time was introduced in \cite{ford1958constructing}.
They also showed how to find a maximum $s$-$t$-flow over time using one minimum-cost flow computation, thus generalizing the concept of static maximum $s$-$t$ flows.
Detailed introductions to flows over time and further references can be found in the surveys by Aronson \cite{aronson1989survey} and Skutella \cite{skutella2009introduction}.

In general, the goal of robust optimization is to find a solution that is feasible and as good as possible for any input data in a given uncertainty set.
For a comprehensive introduction to robust optimization, we refer to \cite{ben2009robust}.
In this paper, we consider $\Gamma$-robustness, a concept introduced by Bertsimas and Sim \cite{BeSi03}.
Here, the uncertainty set is determined by $\Gamma$:
Protection is sought against all scenarios in which the input deviates from the nominal input data in at most $\Gamma$ elements simultaneously.
In the context of static network flows, such a scenario could be the failure of at most $\Gamma$ edges in a given network.
Aneja, Chandrasekaran and Nair \cite{aneja2001maximizing} showed how to solve the $\Gamma$-robust maximum $s$-$t$-flow problem in polynomial time for $\Gamma = 1$, even for the case where the flow is required to be integral.
Du and Chandrasekaran \cite{du2007wronghardness} claimed that the $\Gamma$-robust maximum $s$-$t$-flow problem is NP-hard for $\Gamma >1$.
However, Matuschke et al.\ \cite{2015wronghardnessstatic} recently showed that the proof is incorrect.
Personal communication with Disser and Matuschke \cite{Disser2015newhardnessstatic} indicates that this problem is already NP-hard
if $\Gamma$ is not bounded by any constant.
Still, the complexity of the static $\Gamma$-robust flow problem where $\Gamma$ is bounded by any constant bigger than $1$ is open.
Note that in this model, it is not possible to reroute after the edge failure.

In \cite{bertsimas2013robust}, Bertsimas, Nasrabadi and Stiller investigate several variants of robust flow problems.
As far as we know, \cite{bertsimas2013robust} is the only work that considers robustness in a flow over time setting.
In particular, they study a so-called adaptive model where rerouting after edge-failure is allowed within bounds determined by the initial flow and show that this problem is weakly NP-hard.
In contrast, in this paper, rerouting after edge failures is not permitted and we consider a more general robustness scenario:
Besides total edge failures, the travel time can also increase by a finite amount.

Another problem that is highly related to our work is the network interdiction problem.
In contrast to robust maximum flow problems, where the goal is to find a flow that is good no matter which scenario of the uncertainty set is realized, the network interdiction problem takes the opposite perspective:
Here, the goal is to find a set of edges whose deletion minimizes the amount of flow that can be sent in the remaining network.
Wood \cite{W93} showed NP-hardness of this problem.

Köhler and Skutella \cite{koehler2005} consider a flow over time problem where traffic times depend on the actual load of an edge.
While the resulting flow problem is NP-hard, in contrast to our model, temporally repeated flows can be used to obtain a $2$-approximation.

\sectionheadline{Our Contribution}
First, we discuss generalizations of well-studied concepts for the nominal maximum flow over time problem towards robustness.
For uncertain travel times, we point out problems with s straightforward generalization of encoding flow rates on the edges. We show that this concept as well as time-expanded networks are no longer appropriate.
In contrast, we introduce a non-standard, more viable solution encoding in terms of a path decomposition with associated flow rates and dispatch intervals.
We call this concept \emph{general solutions (to the robust maximum flow over time problem)}.
Whereas these two descriptions are equally powerful in the non-robust, i.e.\ \emph{nominal} flow over time case, this is not true for the robust flow over time model.
We also study a robust counterpart of \emph{temporally repeated flows} due to their simple solution encoding, computational complexity and optimality in the nominal case \cite{skutella2009introduction}.

\newcommand{\scs}{\scriptsize}
\newcommand{\fns}{\footnotesize}
\begin{table}[tb]
 \begin{center}
 \begin{tabular*}{\linewidth}{
@{}l
@{\extracolsep{\fill}} l
@{\extracolsep{3ex}}   l
@{\extracolsep{3ex}}   l
@{\extracolsep{1ex}}   l
}
\toprule
\multirow{4}{*}{\normalsize{Task}}& \multirow{4}{*}{\normalsize{\shortstack[l]{static robust\\ max flow}}} & \multicolumn{3}{c}{\normalsize{robust max flow over time}}\\ \cmidrule{3-5}
&     & \small{general solutions} & \multicolumn{2}{c}{\small{temporally repeated solutions}} \\ \cmidrule{4-5}
       &   &  & \multicolumn{1}{c}{\scs{arbitrary path length}}  & \scs{$T$-bounded path length}  \\
       \midrule
  max, $\Gamma = 1$ & poly time$^1$ \cite{aneja2001maximizing} & ?$^1$ & ? & poly time (Prp. \ref{prop:temp-rep:LP-T-not}) \\[1ex]
  \multirow{2}{*}{max, $\Gamma$ bounded} & \multirow{2}{*}{?} & \multirow{2}{*}{\shortstack[l]{at least as hard as\\ static case$^2$ (Prp. \ref{prop:modelling-techniques:robust-flow-np-hard})}} & \multirow{2}{*}{\shortstack[l]{at least as hard as\\ static case (Prp. \ref{prop:temp-rep:complexity})}} & \multirow{2}{*}{poly time (Prp. \ref{prop:temp-rep:LP-T-not})} \\
  \\[1ex]
 max, $\Gamma$ arb. & NP-hard \cite{Disser2015newhardnessstatic} & NP-hard$^2$ (Prp. \ref{prop:modelling-techniques:robust-flow-np-hard}) & NP-hard (Prp. \ref{prop:temp-rep:complexity}) & poly time (Prp. \ref{prop:temp-rep:LP-T-not}) \\[1ex]
  \multirow{2}{*}{\shortstack[l]{feasibility\\ check}} & \multirow{2}{*}{poly time} & \multirow{2}{*}{NP-hard$^2$ (Thm. \ref{thm:modelling-techniques:separation-np-hard})} & \multirow{2}{*}{poly time} & \multirow{2}{*}{poly time}\\
  \\
  \bottomrule
 \end{tabular*}
 \end{center}
 \caption{An overview of the complexity results from Section \ref{sec:computational-complexity}. \\ \fns{$^1$ If flow rates must be chosen integral, the problem is polynomial time solvable for the static case \cite{aneja2001maximizing} and is inapproximable within any factor for general solutions (see Proposition \ref{prp:SSRFTNE-integral-np-hard-inapproximable}).} \\ \fns{$^2$ These results do not only hold for arbitrary path lengths, but also for instances with the $T$-bounded path length property.}
 }
 \label{tab:results-computational-complexity}
\end{table}

\begin{table}[tb]
 \begin{center}
 \begin{tabular*}{\linewidth}{
@{}l
@{\extracolsep{\fill}} l
@{\extracolsep{5ex}}   l
}
\toprule
\normalsize{delay restriction} & \normalsize{lower bound} & \normalsize{upper bound} \\
\midrule
  $\Delta \in \{0,\infty\}$ & $\max\{T,\Gamma\}$ (Prp. \ref{prop:temp-rep:linear-gap}) & $O(k \log T)$ (Thm. \ref{thm:temp-rep:k-log-T-gap}) \\
  $T$-bounded path length & $\max\{\log T,\log \Gamma\}$ (Prp. \ref{prop:temp-rep:log-gap}) & $O(\eta k \log T)$ (Thm. \ref{thm:temp-rep:delta-arbitrary}) \\
  none & $\max\{T,\Gamma\}$ (Prp. \ref{prop:temp-rep:linear-gap}) & $O(\eta k \log T)$ (Thm. \ref{thm:temp-rep:delta-arbitrary}) \\
\bottomrule
 \end{tabular*}
 \end{center}
 \caption{Outline of the main results regarding optimality gaps of temporally repeated solutions from Section \ref{sec:bounds}.}
 \label{tab:results-bounds}
\end{table}

After introduction and discussion of the models, we address the computational complexity of both solution variants. Our results are summarized in Table \ref{tab:results-computational-complexity}.
The special case of $\Gamma$-robust flow over time with zero travel times and infinite delays on all edges reduces to the \emph{static} $\Gamma$-robust flow problem which searches for an optimum static flow which is robust against up to $\Gamma$ edge failures (cf.\ \emph{related work}).
We show in Proposition \ref{prop:modelling-techniques:robust-flow-np-hard} that our problem contains the static case.
Irrespective of this, our problem for general solutions is considerably more complex.
We show by a reduction from maximum clique that even the verification of feasibility of arbitrary solution candidates is an NP-hard problem (Theorem \ref{thm:modelling-techniques:separation-np-hard}).
By a reduction from the two edge-disjoint paths problem we show that, again in contrast to the static version, the optimization problem is inapproximable for $\Gamma = 1$, if flow rates are required to be integral (Proposition~\ref{prp:SSRFTNE-integral-np-hard-inapproximable}).

For temporally repeated flows, we observe that the computational
complexity depends on the actual edge delays.
In general, they inherit the same complexity status as general solutions (Proposition \ref{prop:temp-rep:complexity}).
However, if the maximum possible path length of each path is bounded by the time horizon, the status changes.
In fact, we show that in this case optimal temporally repeated flows can be computed in polynomial time (Proposition \ref{prop:temp-rep:LP-T-not}).
We formalize this concept in Section \ref{sec:computational-complexity} and call it the \emph{$T$-bounded path length} property.
Note that our hardness results for general solutions also hold for instances with the $T$-bounded path length property.

Subsequently, we study the quality of temporally repeated solutions, when compared to general solutions (Table \ref{tab:results-bounds}).
Temporally repeated solutions do not only have benefits with respect to the computational complexity, but also have a simple solution encoding and are very well studied for the nominal case \cite{skutella2009introduction}.
In the nominal case, it is well known that temporally repeated flows
yield optimal solutions.
Under uncertainty, however, we show that temporally repeated flows can inherit a large optimality gap.

If the delay on all edges is either infinitely large or zero, i.e.\ $\Delta \in \{0,\infty\}$, we show that the optimality gap can be as large as $\max\{T,\Gamma\}$ (Proposition \ref{prop:temp-rep:linear-gap}).
By using a non-trivial primal-dual fitting approach, we show that it is always upper bounded by $O(k \log T)$ (Theorem \ref{thm:temp-rep:k-log-T-gap}).
Here, $k$ is a parameter that is specific for the instance. It does
not depend on the delay but only on the graph structure, travel times
and the time horizon (The definition of $k$ is formalized in Section \ref{sec:bounds}).
Although the instances used in Proposition \ref{prop:temp-rep:linear-gap} have
$k = T$, other classes of instances exist for which the value of $k$
is very small. For example, acyclic digraphs with the $T$-bounded path length property have $k = 1$.

For instances with the $T$-bounded path length property, we provide lower bound examples with an optimality gap of $\max\{\log T,\log \Gamma\}$ (Proposition \ref{prop:temp-rep:log-gap}).
For arbitrary delays, we prove an upper bound of $O(\eta k \log T)$ (Theorem \ref{thm:temp-rep:delta-arbitrary}).
Here, $\eta$ accounts for the relative amount of flow that can be destroyed by scenarios on single paths in the worst case.
Again, this parameter that is again characteristic for an instance is formalized in Section \ref{sec:bounds}.
This bound is tight for the instance from Proposition \ref{prop:temp-rep:log-gap} as, in this case, $k = \eta = 1$.

Finally, if we fix a graph with travel times and finite delays and let the time horizon tend to infinity, we observe that temporally repeated solutions tend to optimality (Proposition \ref{prop:temp-rep:asymptotic-bound}).

\sectionheadline{Outline}
The remainder of this paper consists of three main sections.
Section~\ref{sec:modelling-techniques} is devoted to modelling techniques for the robust maximum flow over time problem.
It introduces a model and two solution concepts which we study subsequently.
Section~\ref{sec:computational-complexity} provides insight in the computational complexity of both concepts under several perspectives.
Finally, Section~\ref{sec:bounds} evaluates the solution quality of optimum temporally repeated flows with respect to their optimality gap to general solutions.
The paper is concluded with final remarks and a discussion of open questions in Section~\ref{sec:conclusions}.

\section{Modeling Techniques for the Robust Maximum Flow over Time Problem}
\label{sec:modelling-techniques}

\sectionheadline{Nominal Maximum Flow over Time Problem}
An instance of the nominal maximum flow over time problem consists of a directed graph $G = (V,E)$ with source and destination vertices $s,d \in V$ and a time horizon $T \in \NN$.
Each edge is equipped with a capacity $u: E \rightarrow \NN$ and a travel time $\tau: E \rightarrow \NN$.
Flow entering edge $e$ at some time $\theta$ leaves the edge at its head at time $\theta + \tau_e$.
We seek to maximize the total amount of flow sent from $s$ to $d$ within the time horizon.
In particular, we use the so-called continuous time model. For further discussion of the relationship between continuous and discrete time models, we refer to \cite{koch2011}.

In classical flow theory, a solution is encoded by a set of Lebesgue-integrable functions $f_e: \RR \rightarrow \RR_+$ for all $e \in E$
describing the rate of flow entering edge $e$ at time $\theta \in \RR$. We assume $f_e(\theta) = 0$ for all $e \in E$ and all $\theta \in \RR \setminus [0,T)$.
A feasible solution obeys the capacity limit, that is, $f_e(\theta) \leq u_e$ for all edges and for all $\theta \in [0,T)$.
Depending on the situation to be modeled, waiting at intermediate vertices may or may not be allowed.
For example, if vertices do not have any buffer capacity, it is impossible to store flow at intermediate vertices.
In such situations, \emph{strict} flow conservation is required, that is, the total amount of flow leaving a
vertex up to any point in time $\xi$ is exactly the total amount of flow entering that vertex up to the same time. We have
\[\sum_{e \in \delta^-(v)} \int_0^{\xi - \tau_e} f_e(\theta) d\theta = \sum_{e \in \delta^+(v)} \int_0^{\xi} f_e(\theta) d\theta\]
for all $\xi \in [0,T)$ and for all $v \in V \setminus \{s,d\}$.
The objective function value of a flow $f$ is defined as the total amount of flow reaching vertex $d$ up to time $T$, that is, \smash{$\sum_{e \in \delta^-(d)}\int_0^{T - \tau_e} f_e(\theta) d\theta$}.
Here, we assume that $d$ has no outgoing edges, otherwise, we would have to subtract the amount of flow leaving $d$ up to time $T$.
We make this assumption throughout this paper. Note that flow that
arrives at $d$ after $T$ is allowed but does not give any contribution
to the value of the objective function.

Ford and Fulkerson \cite{ford1958constructing} showed that relaxing strict flow conservation to weak flow conservation, where flow may wait at intermediate vertices, does not change the optimal objective function value in the nominal case.
That is, they proved that there always exists an optimal flow over time which never stores flow at intermediate vertices.

\sectionheadline{$\Gamma$-Robust Flows over Time}
In this paper, we assume that the travel times $\tau_e$ are uncertain and may deviate by a certain delay $\Delta_e \in \NN$.
We follow the $\Gamma$-robust approach suggested by Bertsimas and Sim \cite{BeSi03}:
for a given integer $\Gamma \in \NN$, we look for a maximum flow over time that is robust against any possible scenario of up to $\Gamma$ edge delays.
This situation can be interpreted as a two-player game in which the first player (the "flow player") decides on a flow over time.
Afterwards, the second player (the "bad adversary") chooses at most $\Gamma$ edges on which she delays the travel times from $\tau_e$ to $\tau_e+\Delta_e$.
The adversary's goal is to minimize the total throughput, or to violate the capacity constraints.
In the robust setting, it might well be that a nominal flow without waiting times has to wait in certain scenarios.
In situations with implied strict flow conservation, such a flow is no longer feasible.
In this paper, we do not consider such situations, and therefore demand strict flow conservation no matter how the adversary chooses edges to be delayed.
We remark that requiring weak flow conservation only introduces additional challenges, as it would have to be decided which flow particles need to wait and which may pass on.

Furthermore, we do not require the network to be empty at time $T$, i.e.\ we allow flow to enter destination $d$ after time $T$. The flow arriving at $d$ after $T$ is not counted in the objective function.
This yields some freedom in reacting to different possible scenarios.
Otherwise, it would be necessary to ensure additionally that - no matter which edges the adversary chooses to delay - all flow can reach $d$ by time $T$.
This would result in a very restrictive model, where we would not be allowed to use any edges whose travel time may exceed $T$. Note that the capacity constraints and the flow conservation constraints are only present up to time $T$ for the same reason. For example, we also allow a feasible solution to violate the edge capacities after $T$.

\sectionheadline{Mathematical Model}

A solution formulated in terms of flow rates $f_e(\theta)$ seems to be unsuitable.
There, the actual flow rate at any intermediate edge depends on the corresponding scenario.
As we are interested in solutions that can be described independently from the scenario, we formulate solutions in terms of $s$-$d$-paths.
A nominal flow over time satisfying strict flow conservation can always be described by a path decomposition $(f_P)_{P\in \mathcal{P}}$, where $\mathcal{P}$ is the set of
$s$-$d$-paths and $f_P: \RR \to \RR_+$ with $f_P(\theta) = 0$ for all $\theta \in \RR \setminus [0,T-\tau(P))$.
For simplicity reasons, in this paper we only allow flow on simple paths.
The function assigns a flow rate $f_P(\theta)$ to each point in time $\theta$ that describes the rate at which flow is sent into path $P$.
Note that this path decomposition is not necessarily unique.
Already in the static robust maximum flow problem, Bertsimas et al.\ have shown that the robust flow value depends on the path decomposition \cite{bertsimas2013robust}.
In this work, however, we do not decompose a flow, but rather define it on the set of all simple paths.

Let $\scenarios = \{z \in \{0,1\}^{|E|} : \sum_{e\in E} z_e \leq \Gamma\}$ denote the set of admissible scenarios, that is, representing all combinations of at most $\Gamma$ delays.
Here, $z_e$ is a decision variable with $z_e = 1$ if and only if $e$ is delayed.
We call a flow over time $(f_P)_{P\in \mathcal{P}}$ \emph{feasible} if the capacity constraints are obeyed under every possible scenario $z\in \scenarios$.
That is, if
\[ \sum_{P\in \mathcal{P}: e\in P} f_P(\theta_{z,e}) \le u_e \quad \forall e\in E, \theta \in [0, T), z\in \scenarios, \]
where
\[\theta_{z,e}=\theta-\sum_{e'\in E: e' <_P e} (\tau_{e'}+\Delta_{e'} z_{e'})\]
denotes the departure time at $s$ of a flow particle on path $P$ which enters edge $e$ at time $\theta$ under scenario $z$.
As usual, $\{e'\in E: e' <_P e\}$ denotes the set of edges that need to be traversed on path $P$ before edge $e$ is reached.
Note that negative values of $\theta_{z,e}$ imply that these flow particles would have started at $s$ before time zero (which is not possible).

The \emph{robust value} of a feasible flow over time $f=(f_P)_{P\in \mathcal{P}}$ is
\[\min_{z\in \scenarios} \sum_{P\in \mathcal{P}} \int_0^{\max\{0,T - \tau(P) - \Delta_z(P)\}} f_P(\nu) d\nu, \]
where $\tau(P) + \Delta_z(P) = \sum_{e\in P} (\tau_e + \Delta_e z_e)$ denotes the travel time on path $P$ in scenario $z$.

\begin{remark}
The nominal maximum flow over time problem, that corresponds to $\Gamma = 0$, can be modeled by an auxiliary time-expanded network.
This graph contains a vertex $(v, t)$ for each relevant time $t \in\{0, \ldots, T - 1\}$.
Edges between the vertices $(v,t)$ and $(w,t + \tau_{vw})$ model edges between vertices $v$ and $w$ with travel time $\tau_{vw}$ in the original graph.
This is not possible in the robust version, since each scenario would induce a different time-expanded network.
Thus, the choice of $z$ defines the topology of the network, and a straight-forward usage of time-expanded networks cannot be applied here.
\end{remark}

By an averaging argument, we show next that it is sufficient to consider functions $f_P$ which are piecewise constant on all integral unit length intervals.

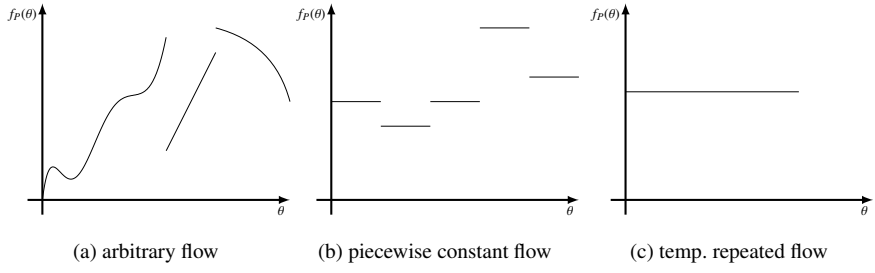
\begin{figure}[t]
\centering
    \begin{subfigure}[b]{0.32\textwidth}  \centering
      \resizebox{\textwidth}{!}{
  \begin{tikzpicture}
\draw[fill=white,draw=white,use as bounding box] (-0.75,-0.5) rectangle (5,4);
         \draw [->,>=latex,very thick] (-0.3,0)--(5,0) node[below left] {$\theta$};
         \draw [->,>=latex,very thick] (0,-0.3)--(0,4) node[below left] {$f_P(\theta)$};
        \draw[smooth,samples= 200,domain=0.0:2.5]
            plot(\x,{8*\x-32.4*\x^2+53.48*\x^3-42.11*\x^4+17.594*\x^5 -3.99*\x^6+0.465713*\x^7-0.0217374*\x^8});
        \draw (2.5, 1) -- (3.5, 3);
        \draw[bend left] (3.5, 3.5) to (5, 2);
\end{tikzpicture}}
  \caption{arbitrary flow}
  \label{fig:arbitrary_flow}
 \end{subfigure}%
\begin{subfigure}[b]{0.32\textwidth}  \centering
  \resizebox{\textwidth}{!}{
  \begin{tikzpicture}
\draw[fill=white,draw=white,use as bounding box] (-0.75,-0.5) rectangle (5,4);
         \draw [->,>=latex,very thick] (-0.3,0)--(5,0) node[below left] {$\theta$};
         \draw [->,>=latex,very thick] (0,-0.3)--(0,4) node[below left] {$f_P(\theta)$};
         \draw (0,2) to (1,2);
         \draw (1,1.5) to (2,1.5);
         \draw (2,2) to (3,2);
         \draw (3,3.5) to (4,3.5);
         \draw (4,2.5) to (5,2.5);

\end{tikzpicture}}
  \caption{piecewise constant flow}
  \label{fig:triple_flow}
 \end{subfigure}
 \begin{subfigure}[b]{0.32\textwidth}  \centering
  \resizebox{\textwidth}{!}{
  \begin{tikzpicture}
\draw[fill=white,draw=white,use as bounding box] (-0.75,-0.5) rectangle (5,4);
         \draw [->,>=latex,very thick] (-0.3,0)--(5,0) node[below left] {$\theta$};
         \draw [->,>=latex,very thick] (0,-0.3)--(0,4) node[below left] {$f_P(\theta)$};
         \draw (0,2.2) to (3.5,2.2);
\end{tikzpicture}}
  \caption{temp.\ repeated flow}
  \label{fig:temp_rep_flow}
 \end{subfigure}
 \caption{Comparison of different models for flows over time.}
 \label{fig:comparision_between_models}
\end{figure}

\begin{restatable}{proposition}{prpssrftpiecewiseconstant}
\label{prop:modelling-techniques:piecewise-constant}
For every solution to the robust maximum flow over time problem there exists a solution with the same objective function value which consists only of piecewise constant functions $f_P$
whose values change only at integer points.
\end{restatable}

\begin{proof}
 Let $\tilde{f}$ be a solution to the robust maximum flow over time problem.
 For every $s$-$d$-path $P$, we construct the piecewise constant $f_P: \RR \rightarrow \RR_+$ as follows.
 Intuitively, cut $\tilde{f}_P$ into unit-intervals $[a,a+1) \subseteq [0,T)$.
 This is possible due to the fact that $T$ is integral.
 Set the flow rate entering path $P$ during each interval to be the average observed in the unit interval, that is,
 $$f_P(\theta) = \int\limits_{a}^{a+1}{\tilde{f}_P(\nu) d\nu} \text{ for all } \theta \in [a,a+1),$$
 and zero outside of the interval $[0,T)$, see Figures \ref{fig:arbitrary_flow} and \ref{fig:triple_flow} for an illustration.
 This ensures that $f$ and $\tilde{f}$ have the same objective function value since all limits of the intervals are integral.

 We now argue that all constraints are satisfied.
 Suppose there is a scenario $z \in \scenarios$ and a time $\theta \in [a,a+1)$ when $f$ violates the capacity of some edge $e \in E$.
 Then the capacity must be violated for all $\theta \in [a,a+1)$ since the solution consists of piecewise constant functions only, and all travel times $\tau_e$ and $\Delta_e z_e$ are integral.
 Therefore, the solution sends more than $u_e$ through this edge in the unit interval.
 We can conclude that the original solution sends more than $u_e$ through this edge in the unit interval, too.
 So there has to be a time $\theta^* \in [a,a+1)$ for which the capacity is violated in the original solution, which contradicts the feasibility of $\tilde{f}$.
  \qed
\end{proof}

Hence, it is sufficient to consider solutions that can be described by a family of triples $\left\{\left(P^i,f^i, [a^i,b^i)\right)\right\}_{i = 1,\dots,\omega}$ that describe
the rate $f^i$ at which flow is sent into path $P^i$ during the time interval $[a^i,b^i)$, with $a_i,b_i\in \NN$.
We call this interval the \emph{dispatch interval}.
Using Proposition \ref{prop:modelling-techniques:piecewise-constant} we redefine a solution to the robust maximum flow over time problem as follows.

\begin{definition}
\label{def:triples}
A solution to the robust maximum flow over time problem is encoded as a family of triples $\left\{\left(P^i,f^i, [a^i,b^i)\right)\right\}_{i = 1,\dots,\omega}$.
For each $1 \leq i \leq \omega$, $P^i$ is a simple $s$-$d$-path in $G$, $f^i > 0$ and $a^i,b^i \in \{0,\dots,T\}$ with $a^i < b^i$.
\end{definition}

\sectionheadline{Temporally Repeated Flows}
Temporally repeated flows are a classical solution concept used to solve the nominal maximum flow over time problem to optimality.
They are constructed from a path decomposition $x = \sum_{P\in \mathcal{P}} x_P$ of a given static flow $x$ by sending flow at rate $x_P$ along a path $P$ as long as the flow can arrive at the sink $d$ by time $T$.
Therefore, a temporally repeated flow represented in our model has dispatch intervals of the form $[0,  T-\tau(P))$ for a path $P$ and flow rates $f_P = x_P$.
A temporally repeated flow is called \emph{feasible}, if
\begin{align}\label{eq:temprep} \sum_{i : e \in P^i} f^i \leq u_e \text { for all edges } e \in E,\end{align}
that is, the capacity constraints are satisfied independent of the actual point in time.
Note that this will always be the case, if the temporally repeated flow was constructed from a \emph{feasible} static flow.

Equivalently, we can state that a flow $\big\{(P^i, f^i, [a^i, b^i))\big \})$ is a feasible temporally repeated flow if and only if $a^i = 0$ and $b^i \ge T - \tau(P^i)$
for all $i$ and \eqref{eq:temprep} holds.
Since the dispatch intervals are fixed by the paths, we will omit them from now on.
See Figure~\ref{fig:comparision_between_models} for a graphical comparison between the different flow models.

\section{Computational Complexity of the Robust Maximum Flow over Time Problem}
\label{sec:computational-complexity}

\sectionheadline{General Solutions}
The following proposition shows that the robust maximum flow over time problem is at least as hard as the static counterpart.
Recently, Disser and Matuschke \cite{Disser2015newhardnessstatic} disclosed that they were able to prove NP-hardness of the static counterpart for unbounded $\Gamma$.
The result, however, is not yet published.

\begin{restatable}{proposition}{SSRFTNNEnphard}
 \label{prop:modelling-techniques:robust-flow-np-hard}
The robust maximum flow over time problem is at least as hard as the static robust maximum flow problem.
\end{restatable}

\begin{proof}
 We will show that the robust maximum flow over time problem can be
 used to solve the static robust maximum flow problem.

 Let us assume an instance $(G = (V,E),s,d,u,\Gamma)$ of robust maximum flow is given.
 We construct an instance $(G = (V,E),s,d,u,\tau,\Delta,T,\Gamma)$ of robust maximum flow over time as follows.
 Let $T = 1, \tau \equiv 0$ and $\Delta \equiv \infty$.
 Then any feasible solution to the robust maximum flow problem can be mapped to a solution $\left\{\left(P^i,f^i, [a^i,b^i)\right)\right\}_{i = 1,\dots,\omega}$ of the robust maximum flow over time instance.
 For every path $P$ with flow rate $f_P$ in the robust maximum flow
 solution, the over-time solution sends the same amount of flow during the dispatch interval $[0,1)$.
 The worst-case scenario of the robust maximum flow destroys at most $\Gamma$ edges and decreases the robust flow value - with respect to the nominal flow value - by the total flow rate summed among all
 paths that are affected by the set of deleted edges.
 The same interference occurs in the constructed instance of robust maximum flow over time.
 If an edge was delayed, the travel time increases to an arbitrarily large value.
 Hence, any path using such an edge does not reach the destination.
 Thus, the objective function values coincide and an optimal solution to the robust maximum flow over time instance is also optimal for the robust maximum flow problem. \qed
\end{proof}

The following two results show that the temporal component introduces additional difficulty to the robust problem, as the static counterparts are both polynomially solvable.
First, we show that computing an optimal integral solution in polynomial time is unlikely and that the problem cannot be approximated.

\begin{proposition}
\label{prp:SSRFTNE-integral-np-hard-inapproximable}
 The robust maximum flow over time problem is NP-hard, if we require integral flow rates, that is, $f^i \in \ZZ$, even if $\Gamma = 1$.
 Moreover, the problem is inapproximable within any factor.
\end{proposition}

\begin{proof}
 We will show that robust maximum flow over time with integral flow rates solves the two edge-disjoint paths problem.
 An instance of two edge-disjoint paths consists of a graph $G = (V,E)$ with two designated pairs of vertices $(s_i,d_i), i = 1,2$ and asks for two edge-disjoint $s_i$-$d_i$-paths in $G$.
 We assume that $s_1,s_2,d_1,d_2$ are pairwise disjoint.
 The problem was shown to be NP-hard by Fortune et al.\ \cite{fortune1980directed}.

 We model this task as a robust maximum flow over time problem as follows (see Figure \ref{fig:SSRFTNE-integral-np-hard}).
 We introduce two auxiliary vertices $s$ and $d$ and connect $s$ to $s_i$ and $d$ to $d_i$.
 Hence, we end up with the graph $\bar{G} = (V \cup \{s,d\}, E \cup \{s,s_1\} \cup \{s,s_2\} \cup \{d_1,d\} \cup \{d_2,d\})$.
 Let $u \equiv 1, \Delta \equiv 2, \tau_e = 0$ for all $e \in E$, $\tau_{\{s,s_1\}} = 0, \tau_{\{s,s_2\}} = 1, \tau_{\{d_1,d\}} = 1, \tau_{\{d_2, d\}} = 0$ and $T = 2$.

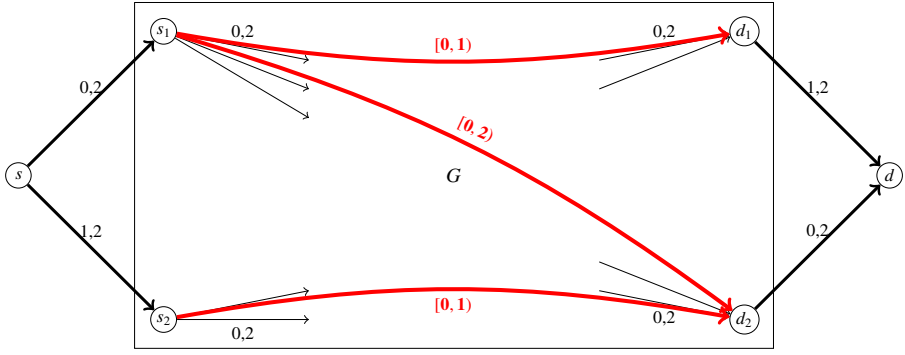
\begin{figure}[tb]
 \centering
 \tikzstyle{vertex}=[circle, draw,
                        inner sep=1pt, minimum width=10pt]
\tikzstyle{edge} = [draw,very thick]

    \resizebox{\textwidth}{!}{
    \begin{tikzpicture}[scale = 2]

    \node[vertex] (s) at (0, 1){\scriptsize $s$};
    \node[vertex] (s1) at (1, 2){\scriptsize $s_1$};
    \node[vertex] (s2) at (1, 0){\scriptsize $s_2$};

    \node[vertex] (d) at (6, 1){\scriptsize $d$};
    \node[vertex] (d1) at (5, 2){\scriptsize $d_1$};
    \node[vertex] (d2) at (5, 0){\scriptsize $d_2$};

    \draw (s) edge[->, very thick] node[anchor = south] {\scriptsize 0,2}  (s1);
    \draw (s) edge[->, very thick] node[anchor = south] {\scriptsize 1,2}  (s2);

    \draw (d1) edge[->, very thick] node[anchor = south] {\scriptsize 1,2}  (d);
    \draw (d2) edge[->, very thick] node[anchor = south] {\scriptsize 0,2}  (d);

    \draw (s1) edge[->] node[anchor = south] {\scriptsize 0,2} (2,1.8);
    \draw (s1) edge[->] (2,1.6);
    \draw (s1) edge[->] (2,1.4);

    \draw (s2) edge[->] node[anchor = north] {\scriptsize 0,2} (2,0);
    \draw (s2) edge[->] (2,0.2);

    \draw (4,0.2) edge[->] node[anchor = north] {\scriptsize 0,2} (d2);
    \draw (4,0.4) edge[->] (d2);

    \draw (4,1.8) edge[->] node[anchor = south] {\scriptsize 0,2} (d1);
    \draw (4,1.6) edge[->] (d1);

    \draw (0.8,2.2) -- (5.2,2.2) -- (5.2,-0.2) -- (0.8,-0.2) -- (0.8,2.2);
    \node at (3,1) {$G$};

    \draw (s1) edge[->,draw=red,ultra thick,bend right=10] node[anchor=south] {\scriptsize \color{red}{$\mathbf{[0,1)}$}} (d1);
    \draw (s1) edge[->,draw=red,ultra thick,bend left=10] node[anchor=south,rotate=-25] {\scriptsize \color{red}{$\mathbf{[0,2)}$}} (d2);

    \draw (s2) edge[->,draw=red,ultra thick,bend left=10] node[anchor=north] {\scriptsize \color{red}{$\mathbf{[0,1)}$}}  (d2);

    \end{tikzpicture}}
    \caption{Construction for Proposition \ref{prp:SSRFTNE-integral-np-hard-inapproximable}.
Edge labels denote the travel times and possible delays, respectively.
The thick edges depict the possible paths and time intervals a feasible solution may use.}
    \label{fig:SSRFTNE-integral-np-hard}
\end{figure}
 If we solve the robust maximum $s$-$d$-flow over time problem, the travel times ensure that only the following path types can contribute to the objective function value of any feasible solution:
 \begin{enumerate}[1)]
  \item  $s$-$d$-paths traversing $s_1$ and $d_1$,
  \item $s$-$d$-paths traversing $s_1$ and $d_2$,
  \item $s$-$d$-paths traversing $s_2$ and $d_2$.
  \end{enumerate}
 Let us assume that there are two edge-disjoint paths $P_1$ and $P_2$ in $G$, that is, a YES-instance.
 Then we can construct a solution to robust maximum flow over time which sends flow along ${s,s_1},P_1,{d_1,d}$ in the interval $[0,1)$ and along ${s,s_2},P_2,{d_2,d}$ in the interval $[0,1)$.
 Since both paths are disjoint, the total amount of flow reaching the destination is two.
 Moreover, no scenario can destroy more than one of the two paths, hence, the robust flow value is equal to one.

 Now, let us assume a NO-instance of two edge-disjoint paths.
 Due to the integrality of the dispatch intervals and flow rates, a solution to robust maximum flow over time in this network can only consist of at most two different paths.
 If it uses a path of type 1) and a path of type 3), both paths have to share some common edge $e^{*} \in E$
 as we assumed to have a NO-instance of two edge-disjoint paths.
 The scenario delaying exactly this edge reduces the flow value to zero.
 In all other cases, the worst case scenario may delay the single utilized edge leaving $s$ (resp.\ entering $t$) in order to decrease the robust flow value to zero.

 Hence, the instance is a NO-instance if and only if the robust flow value is zero.
 Note that in the construction used above, the objective function value is zero for a NO-instance and one for a YES-instance.
 Hence, any approximation algorithm would still distinguish between YES- and NO-instances. \qed
\end{proof}

The following theorem shows that it is strongly NP-hard to verify feasibility
of robust flow over time solutions in general.

\begin{restatable}{theorem}{thmmodellingtechniquesseparationnphard}
\label{thm:modelling-techniques:separation-np-hard}
Deciding feasibility of a given solution $f = \left\{\left(P^i,f^i,
    [a^i,b^i)\right)\right\}_{i = 1,\dots,\omega}$ is NP-hard.
\end{restatable}

\begin{proof}
We provide a polynomial-time reduction from the \textsc{clique} decision problem, which is one of Karp's classical NP-hard problems \cite{karp1972reducibility}.
We will show that, given any graph $\bar{G}=(\bar{V},\bar{E})$ and some $r \in \NN$, we can construct a maximum robust flow over time instance and a corresponding solution $f$ such that the following holds.
There is a clique of size $r$ in $\bar{G}$ if and only if $f$ is an infeasible solution for the maximum robust flow over time instance.
Without loss of generality, we assume \ $|\bar{E}| \ge |\bar{V}|$, $r\geq 3$ and that the vertices in $\bar{V}$ are numbered from $1, \ldots, n$. Let $m = |\bar{E}|$.

We construct a multigraph $G = (V,E)$, with vertices $V = \{s,d_0,d_1\} \cup \{ v_i^\ell, v_i^r : i \in \bar{V} \}$ as illustrated in Figure \ref{fig:modelling-techniques:separation-np-hard} for the example of $\bar{G}=K_3$.
$V$ consists of two vertices $v_i^\ell$ and $v_i^r$ for each vertex $i \in \bar{V}$, together with three vertices $s$, $d_0$ and $d_1$.
The edge set $E$ consists of five types of edges (omitted values for edges $e \in E$ are $u_e = \infty, \tau_e = 0$ and $\Delta_e = 0$):
\begin{enumerate}[1)]
\item $e = (v_i^\ell,v_i^r)$ for each vertex $i \in \bar{V}$ with $\Delta_e = 2^{i}$,
\item $(v_i^r,d_0)$ for each vertex $i \in \bar{V}$,
\item all ``backward'' edges $(v_i^r,v_j^\ell)$ for $i \neq j\in \bar{V}$,
\item $e = (d_0,d_1)$ with $u_e = \smash{\binom{r}{2}} - 1$,
\item $e = (s,v_i^\ell)$ for each edge $\bar{e} \in \bar{E}$ such that $i \in \bar{e}$.
For $\bar{e} = \{i, j\}$, set  $\tau_e = 2^{m+1} - 2^{i} - 2^{j}$. These edges can be parallel.
\end{enumerate}

We set $T = 2^{m+1} + 1$ and $\Gamma = r$. The solution candidate for the feasibility problem is constructed as follows.
We introduce a triple $(P^{\bar{e}}, f^{\bar{e}}, [0,1))$ for each edge $\{i, j\} = \bar{e} \in \bar{E}$,
where $P^{\bar{e}} = (s, v_i^\ell, v_i^r, v_j^\ell, v_j^r, d_0, d_1)$.
The edge $(s, v_i^\ell)$ is chosen to be the designated edge for $\bar{e}$ of type (5) above.

The following two claims conclude the proof as we will show that a $r$-clique in $\bar{G}$ exists if and only if the constructed solution is infeasible due to a capacity violation at time $t = 2^{m+1}$.

\begin{figure}[tb]
 \centering
 \tikzstyle{vertex}=[circle, draw,
                        inner sep=1pt, minimum width=13pt,fill=white]
\tikzstyle{edge} = [draw,very thick]
    \resizebox{\textwidth}{!}{
    \begin{tikzpicture}[scale = 2]
    \begin{scope}[xshift=-2cm,yshift=0cm]
    \node[vertex] (v1) at (0, 0){\scriptsize $v_1$};
    \node[vertex] (v2) at (1.6, 0){\scriptsize $v_2$};
    \node[vertex] (v3) at (0.8, 1){\scriptsize $v_3$};

    \draw (v1) edge[-] node[anchor=north] {$e_1$} (v2);
    \draw (v1) edge[-] node[anchor=south,xshift={-.2cm}] {$e_2$} (v3);
    \draw (v2) edge[-] node[anchor=south,xshift={.2cm}] {$e_3$} (v3);
    \end{scope}
    \draw (-0.2,-0.15) edge[-,thick] (-0.2,1.15);

    \node[vertex] (s) at (0, .5){\scriptsize $s$};
    \node[vertex] (d0) at (3, .5){\scriptsize $d_0$};
    \node[vertex] (d1) at (4, .5){\scriptsize $d_1$};

    \node[vertex] (v1l) at (1, 0){\scriptsize $v^{\ell}_1$};
    \node[vertex] (v1r) at (2, 0){\scriptsize $v^r_1$};

    \node[vertex] (v2l) at (1, 0.5){\scriptsize $v^{\ell}_2$};
    \node[vertex] (v2r) at (2, 0.5){\scriptsize $v^r_2$};

    \node[vertex] (v3l) at (1, 1){\scriptsize $v^{\ell}_3$};
    \node[vertex] (v3r) at (2, 1){\scriptsize $v^r_3$};

    \draw (v1l) to node[anchor=south] {\scriptsize $\Delta = 2$} (v1r);
    \draw (v2l) to node[anchor=south] {\scriptsize $\Delta = 4$} (v2r);
    \draw (v3l) to node[anchor=south] {\scriptsize $\Delta = 8$} (v3r);

    \draw (d0) to node[anchor=south,yshift={-.05 cm}] {\scriptsize $u = 2$} (d1);

        \draw[draw=red] (s) to[bend right=15] node[anchor=north,rotate=-30] {\scriptsize \color{red}{$P_1$}, $\tau=10$} (v1l);
         \draw[draw=red] (v1l) to[bend right=15] (v1r);
         \draw[draw=red] (v1r) to[bend right=5] (v2l);
         \draw[draw=red] (v2l) to[bend right=15] (v2r);
         \draw[draw=red] (v2r) to[bend right=10] (d0);
         \draw[draw=red] (d0) to[bend right=40] (d1);

    \draw[draw=darkgreen,dashed] (s) to node[anchor=south, rotate=25] {\scriptsize \color{darkgreen}{$P_3$}, $\tau=4$} (v3l);
    \draw[draw=darkgreen,dashed] (v3l) to[bend right=15] (v3r);
    \draw[draw=darkgreen,dashed] (v3r) to[bend right=5] (v2l);
    \draw[draw=darkgreen,dashed] (v2l) to[bend left=30] (v2r);
    \draw[draw=darkgreen,dashed] (v2r) to[bend left=10] (d0);
    \draw[draw=darkgreen,dashed] (d0) to[bend right=15] (d1);
    \draw[draw=blue,dotted,thick] (s) to[bend left=10] node[anchor=south,rotate=-30] {\scriptsize \color{blue}{$P_2$}, $\tau=6$} (v1l);
    \draw[draw=blue,dotted,thick] (v1l) to[bend left=30] (v1r);
    \draw[draw=blue,dotted,thick] (v1r) to[bend left=15] (v3l);
    \draw[draw=blue,dotted,thick] (v3l) to[bend left=55] (v3r);
    \draw[draw=blue,dotted,thick] (v3r) to (d0);
    \draw[draw=blue,dotted,thick] (d0) to[bend left=45] (d1);

    \end{tikzpicture}}
\caption{Construction from Theorem \ref{thm:modelling-techniques:separation-np-hard} for the example of $\bar{G} = K_3$.
The graph $\bar{G}$ is on the left and the corresponding robust flow over time instance with $T=2^{m+1}+1=17$ and $\Gamma = r$ is on the right hand side.}
\label{fig:modelling-techniques:separation-np-hard}
\end{figure}
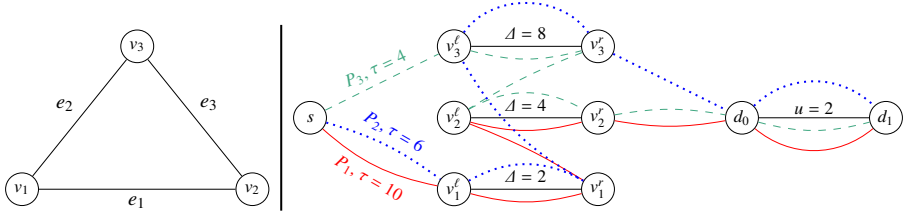

\omsClaim{The constructed solution obeys the capacity constraints for every point in time $t < 2^{m+1}$.}
For the proof, observe that only the edge $(d_0,d_1)$ has finite capacity and that only edges $(v_i^{\ell}, v_i^r)$ can be delayed.
Hence, it is sufficient to argue that its capacity cannot be exceeded unless $t = 2^{m+1}$.
For any point in time $t < 2^{m+1}$, a path can only contribute to the capacity violation if it was delayed at most once.
If it had been delayed twice, it would have had a total travel time of $2^{m+1}-2^i-2^j+2^i+2^j = 2^{m+1}$ by construction.

Now, let us consider some edge $(v_i^\ell,v_i^r)$ and assume that it was delayed.
Then all paths which cross that edge and have no other delayed edge, will have a total travel time of $2^{m+1} - 2^j$ for some $j \neq i$.
By construction, at most $\Gamma = r$ many edges are delayed.
We conclude that at most $r$ paths can have the same travel time, which in particular, are strictly fewer than $\binom{r}{2}$ paths for $r \geq 3$.

We now observe that for $\{i, j\} \not = \{i', j'\}$, $2^i + 2^j \not = 2^{i'} + 2^{j'}$ and that $2^i \not = 2^{i'} + 2^{j'}$ for all $i, i', j'$.
Therefore, less than  $\binom{r}{2}$ paths can arrive at $d_0$ at any point in time before $2^{m+1}$  and thus, the capacity cannot be violated.

\omsClaim{There is a scenario in which the solution violates the capacity constraint for $t = 2^{m+1}$ and edge $e = (d_0,d_1)$ if and only if there is a clique of size $r$ in $\bar{G}$.}
In order to prove the claim, we will show that a clique $C$ of size $r$ in $\bar{G}$ proves that the solution $f$ is infeasible.
Let us select a scenario $z \in \scenarios$ with $z_{(v_i^\ell,v_i^r)} = 1$ if and only if $i \in C$.
Now, for each path $P^{\bar{e}}$ there are three possible situations:
\begin{enumerate}
\item The path contains no delayed edge, then $\tau(P^{\bar{e}}) = 2^{m+1} - 2^u - 2^v < 2^{m+1}$.
  Hence, it does not contribute to the capacity of edge $(d_0,d_1)$ at time $2^{m+1}$.
\item If the path is delayed exactly once, the same argument holds and $\tau(P^{\bar{e}}) < 2^{m+1}$.
\item Finally, if the path is delayed exactly twice, $\tau(P^{\bar{e}}) = 2^{m+1}$ and the path contributes to the capacity.
\end{enumerate}
Consequently, a path $P^{\bar{e}}$ contributes to the capacity violation at edge $(d_0,d_1)$ if and only if it was delayed twice.
Due to the construction of $z$, this is the case if and only if both endpoints of $\bar{e}$ and thus $\bar{e}$ itself was part of the clique.
Since the clique consists of $\binom{r}{2}$ edges, the edge capacity is violated.

On the other hand, if the solution is infeasible, there exists a scenario $z \in \scenarios$ such that at least $\binom{r}{2}$ paths are delayed exactly twice.
Since $\Gamma = r$, this is only possible if the delayed paths form a $r$-clique in $\bar{G}$. \qed
\end{proof}

\sectionheadline{Temporally Repeated Flows}
Next, we discuss the computational complexity of temporally repeated flows, which were introduced in Section \ref{sec:modelling-techniques}.
We observe that computing optimal temporally repeated flows is an NP-hard task in general.

\begin{restatable}{proposition}{proptemprepcomplexity}
\label{prop:temp-rep:complexity}
  In general, the problem of computing an optimal robust temporally repeated flow is at least as hard as the static robust maximum flow problem.
\end{restatable}

\begin{proof}
We consider again the construction in the proof of Proposition~\ref{prop:modelling-techniques:robust-flow-np-hard}.
Recall that we had $T = 1$. In Proposition~\ref{prop:modelling-techniques:piecewise-constant}, we showed
that it suffices to consider flows with integer dispatch intervals.
Therefore, the dispatch interval for any path is $[0, 1)$ and there is an optimal flow over time in this network that is a temporally repeated flow. \qed
\end{proof}

There are situations in which an optimal temporally repeated flow can be computed efficiently.
The following Proposition \ref{prop:temp-rep:LP-T-not} considers instances whose longest robust path length does not exceed the time horizon for any possible scenario $z \in \scenarios$.
That is, we say an instance has the \emph{$T$-bounded path length property} if and only if
\[\max_{P \in \paths, z \in \scenarios} \{ \tau(P) + \Delta_z(P) \} \leq T.\]
This ensures that we can write down an LP model for the problem whose pricing problem is tractable. That means, we can solve the dual separation problem efficiently: given a dual solution,
decide whether the solution is feasible and, if not, return a violated dual inequality. This concludes that the LP can also be solved in polynomial time due to \cite{grotschel1981ellipsoid}.

\begin{restatable}{proposition}{proptemprepTnot}
\label{prop:temp-rep:LP-T-not}
 An optimal robust temporally repeated flow can be computed in polynomial time for instances with the $T$-bounded path length property.
\end{restatable}

\begin{proof}
For ease of notation, we denote the set of all feasible temporally repeated flows by $\mathcal{F} = \{ x \in \RR_+^{|\mathcal{P}|} : \sum_{P : e \in P} x_P \leq u_e, \forall e \in E \}$.
 We can formulate the problem as follows:
 \begin{align*}
  \max_{x \in \mathcal{F}} \min_{z \in \scenarios} \left\{ \sum_{P \in \paths} \left(T - \tau(P) - \Delta_z(P) \right) x_P \right\}
 \end{align*}
 where, as before, $\scenarios=\{z\in \{0,1\}^E\mid \sum_{e\in E} z_e \le \Gamma\}$, and $\Delta_z(P)=\sum_{e\in P} z_e \Delta_e$.
 The variables $x_P$ denote the flow rate at which flow is sent into
 path $P \in \paths$ and the $z_e$ variables model the decision on
 which edge the travel time increases.
 Note that the additional assumption on the maximal path length makes sure that the objective coefficients of all paths in all scenarios can never be negative.

 Each flow $x\in \mathcal{F}$ induces a load of $x_e=\sum_{P\in \paths:e\in P}x_P$ on each edge $e\in E$.
 Thus, the term $\sum_{P\in \paths} \Delta_z(P) x_P$ can be rewritten as $\sum_{e\in E} \Delta_e z_e x_e$.
 As a consequence we can reformulate the objective function of the problem above as:
 $$\max_{x\in \mathcal{F}}\left( \sum_{P \in \paths} \left(T - \tau(P)\right) x_P - \max_{z\in \scenarios}\sum_{e\in E} \Delta_e z_e x_e \right).$$

 Let us consider the inner problem $\max_{z\in \scenarios}\sum_{e\in E} \Delta_e z_e x_e$ for a fixed flow $x\in \mathcal{F}$.
 Since the linear inequality system $\{ \sum_{e\in E} z_e \le \Gamma,\ 0\le z_e \le 1 \forall e\in E \}$ is totally unimodular,
 we might as well replace $\max\{\sum_{e\in E} \Delta_e z_e x_e : z\in \scenarios\}$ by its linear relaxation
 \begin{align*}
 \max_{z\in [0,1]^{|E|}} \left\{ \sum_{e\in E} \Delta_e z_e x_e: \sum_{e\in E} z_e \le \Gamma \right\}.
 \end{align*}
By strong LP duality (using $\gamma_0$ as dual variable for the GUB constraint and $\gamma_e$ for the upper variable bounds), the objective function value of this inner problem coincides with the objective function value of the dual problem
 \begin{align*}
 \min_{\gamma \in \RR_+^{|E|+1}}\quad & \gamma_0 \Gamma + \sum_{e \in E} \gamma_e \\
 \text{s.t.} \quad & \gamma_0 + \gamma_e \geq \Delta_e \sum_{P\in \paths : e \in P} x_P\quad && \forall e \in E.
 \end{align*}
As a consequence, we can reformulate the problem to compute an optimal temporally repeated flow as:
 \begin{align*}
  \max_{x \in \RR_+^{|\mathcal{P}|},\ \gamma \in \RR_+^{|E|+1}} \quad & \sum_{P \in \paths} x_P (T - \tau(P)) - \gamma_0 \Gamma - \sum_{e \in E} \gamma_e\\
 \text{s.t.} \quad & \sum_{P\in \paths : e \in P} x_P \leq u_e \quad && \forall e \in E \\
 & \gamma_ 0 + \gamma_e \geq \Delta_e \sum_{P\in \paths : e \in P} x_P \quad && \forall e \in E.
 \end{align*}
We denote the dual variables for capacity constraints by $\alpha \in \RR_+^{|E|}$ and the dual variables for scenarios by $\beta \in \RR_+^{|E|}$ and obtain the following dual:
 \begin{align*}
 \min_{\alpha,\ \beta \in \RR_+^{|E|}} \quad & \sum_{e \in E} \alpha_e u_e \\
 \text{s.t.}\quad & \sum_{e \in P} \alpha_e + \sum_{e \in P} \beta_e \Delta_e \geq T - \sum_{e \in P} \tau_e\quad && \forall P \in \mathcal{P} \\
 & \sum_{e \in E} \beta_e \leq \Gamma \\
 & \beta_e \leq 1\quad && \forall e \in E.
 \end{align*}
The pricing problem for path variables $x_P$ corresponds to separating the inequalities $\sum_{e \in P}(\alpha_e + \beta_e \Delta_e) \geq T - \sum_{e \in P} \tau_e$ for all paths $P \in \paths$.
 By moving $\sum_{e \in P} \tau_e$ to the left hand side, it can be solved as a shortest path problem with cost $c_e = \alpha_e + \beta_e \Delta_e + \tau_e$, which can be solved in polynomial time.
 Since the ellipsoid method requires only a polynomial time separation oracle in order to run in polynomial time \cite{grotschel1981ellipsoid}, this concludes that the LP can also be solved in polynomial time.
 \qed
\end{proof}

Note that, in general, checking whether the longest path length exceeds the time horizon is NP-hard.
Still, whenever the time horizon is sufficiently large, the $T$-bounded path property is certainly fulfilled.

\section{Bounds on the Solution Quality of Temporally Repeated Flows}
\label{sec:bounds}

In the remainder of the paper, we turn our focus on the solution quality of temporally repeated flows compared to a general solution.
For any given instance $\mathcal{I}$ of the robust maximum flow over time problem, let $f_{\text{TR}}^{\text{OPT}}$ be the value of an optimal temporally repeated and $f^{\text{OPT}}$ be the value of an optimal general solution.
We call the ratio ${f^{\text{OPT}}}/{f_{\text{TR}}^{\text{OPT}}}$ the \emph{optimality gap} of $\mathcal{I}$ and provide general upper and lower bounds on the optimality gap.

\subsection{Lower Bounds}
The next two propositions provide lower bounds on the worst case optimality gap and show that the class of temporally repeated flows is not necessarily optimal for the robust maximum flow over time problem, even for $\Gamma = 1$.
Proposition~\ref{prop:temp-rep:linear-gap} presents a family of instances with optimality gap $\Omega(T + \Gamma)$.
For instances with $T$-bounded path length, the following Proposition~\ref{prop:temp-rep:log-gap} yields a gap of $\Omega(\log T + \log \Gamma)$.
Both families of instances are depicted in Figure~\ref{fig:temp-rep:gap-instances}.

\begin{restatable}{proposition}{proptempreploggap}
 \label{prop:temp-rep:log-gap}
The optimality gap of temporally repeated flows can be $\mathcal{H}_T$ and $\mathcal{H}_{\Gamma+1}$ for instances satisfying the $T$-bounded path length property.
Here, $\mathcal{H}_r \in \Omega(\log{r})$ is the r-th harmonic number, i.e.\ $\mathcal{H}_r = \sum_{i = 1}^r \frac{1}{i}$.
\end{restatable}

\begin{proof}
 For fixed $r \in \NN$, we construct an instance $I_r$ with graph $G_r = (V,E_r)$ as follows (see Figure \ref{fig:temp-rep:gap-instances:log-gap}).
 The vertex set always consists of three vertices $V = \{s,v,d\}$.
 The edge set consists of a designated edge $e^* = (v,d)$ and a set of $r$ parallel edges $e_i = (s,v), i = 0,\dots,r-1$.
 We set $\tau_{e^*} = \Delta_{e^*} = 0$, $\tau_{e_i} = i, \Delta_{e_i} = r - i, \Gamma = r - 1, T = r$.
 All edges have unit capacity.
 For the sake of notation, we use $P^i$ to denote the unique $s$-$d$-path which contains edge $e_{r-i-1}$.
 A combination of the two claims below yields an optimality gap of $\mathcal{H}_r = \Omega(\log r)$.
 Note that $r = \Gamma + 1 = T$.

\omsClaim{There exists a solution to $I_r$ with objective function value $1$.}
 For each $i = 0,\dots,r-1$, we send one unit of flow along path $P^i$ in the time interval $[0,1)$.
 The flow of each path reaches the designated edge $e^*$ during the time interval $[i, i+1)$, or if the path was delayed, during ${[r,r+1)}$ which exceeds the time horizon.
 Since the time intervals do not overlap during the time horizon, there will be no collisions.
 Furthermore, the total flow sent in the nominal setting is equal to the time horizon $T = r$.
 Since the total flow sent along each path is equal to $1$, any scenario can destroy at most $\Gamma = (r-1)$ units of flow, thus leaving a remaining flow of $1$.

 \omsClaim{An optimal temporally repeated flow can send at most $\mathcal{H}_r^{-1}$ units of flow.}
 Let us assume that $x^*$ is an optimal temporally repeated flow with flow rate $x^*_{P^i}$ on path $P^i$.
 For ease of notation, we use $x^*_i = x^*_{P^i}$ in the remainder of the proof.
 Since the flow is temporally repeated, $\sum x^*_i \leq 1$ holds.
 Without loss of generality, we can assume equality.
 Otherwise, we could increase some variable $x^*_i$ by a sufficiently small positive amount $\epsilon$ without decreasing the robust flow value (Note that any scenario can destroy at most the additional flow that would be sent by increasing the variable by $\epsilon$).
 In the following, we claim that the optimal temporally repeated flow is of the following form:
  \[x^*_0 = 2 x^*_1, \quad x^*_1 = \frac{3 x^*_2}{2}, \quad \dots, \quad x^*_{r-2} = \frac{r x^*_{r-1}}{r - 1}. \]
 First, let us assume that the solution is of such form.
 Then, by substituting, we get $x^*_i = \nicefrac{x^*_0}{i + 1}$ and $\sum x^*_i = 1$ implies that $x^*_0 \leq (\sum_{i=1}^r 1/i)^{-1}$.

 Also note that the objective function coefficient of path $P^i$ is $T - \tau_{e_{r - i - 1}} = r - (r - i - 1) = i + 1$ for all $i = 0,\dots,r-1$.
 Hence, $x^*$ satisfies for every $i,j$, $$x^*_i (T - \tau_{e_{r - i - 1}}) = \frac{x^*_0 (i+1)}{i+1} =  x^*_0 = \frac{x^*_0 (j+1)}{j+1} = x^*_j (T - \tau_{e_{r - j - 1}}).$$
 In other words, destroying any set of edges of cardinality $\Gamma$ results in the same objective function value.
 Hence, in every scenario, one of the $r$ paths remains unharmed, yielding an objective function value of $x^*_0 \leq (\sum_i 1/i)^{-1}$.

 Now, let us assume that the solution is not of such form.
 Then, there exist two indices $\ell = \arg\min_{i = 0,\dots,r - 1} x_{i} (T - \tau_{e_{r - i - 1}})$ and $k = \arg\max_{i = 0,\dots,r - 1} x_{i} (T - \tau_{e_{r - i - 1}})$
 with $x_{\ell} (T - \tau_{e_{r - \ell - 1}}) < x_k (T - \tau_{e_{r - k - 1}})$.
 The worst case scenario will delay all edges except for $e_{\ell}$.
 Averaging the flow shows that it was not optimal:
 Decreasing the flow on all edges with weighted flow value equal to $x_k(T - \tau_{e_{r - k - 1}})$ and
 increasing the flow on all edges with weighted flow value equal to $x_\ell (T - \tau_{e_{r - \ell - 1}})$ would strictly increase the robust flow value. \qed
\end{proof}

\begin{figure}[t]
\begin{center}
 \tikzstyle{vertex}=[circle, draw,
                        inner sep=1pt, minimum width=10pt]
\tikzstyle{edge} = [draw,very thick]
\begin{subfigure}[t]{.49\textwidth} \centering
    \begin{tikzpicture}[scale = 1.6]
      \draw[fill=white,draw=white,use as bounding box] (-0.2,-0.8) rectangle (3.2,0.8);

      \node[vertex] (v0) at (0, 0){ $s$};
      \node[vertex] (v1) at (2, 0){ $v$};
      \node[vertex] (v2) at (3, 0){ $d$};

      \node[anchor=north east] at (3.2, 0.8){$T = 3, \Gamma = 2$};

      \draw (v1) edge[->, very thick] node[anchor = south] { 0,0}  (v2);
      \draw (v0) edge[->, very thick, bend left = 55] node[anchor = south] { 0,3}  (v1);
      \draw (v0) edge[->, very thick] node[anchor = south] { 1,2}  (v1);
      \draw (v0) edge[->, very thick, bend right = 55] node[anchor = north] { 2,1}  (v1);
    \end{tikzpicture}
    \caption{\parbox{.9\textwidth}{Instance $I_3$ from Proposition \ref{prop:temp-rep:log-gap}. Edge labels denote the travel times and possible delays, respectively.}}
    \label{fig:temp-rep:gap-instances:log-gap}
\end{subfigure}
\begin{subfigure}[t]{.5\textwidth} \centering
    \begin{tikzpicture}[scale = 1.6]

      \node[vertex] (v0) at (0, 0){ $s$};
      \node[vertex] (v1) at (1.2, 0){ $v_1$};
      \node[vertex] (v2) at (2.2, 0){ $v_2$};
      \node[vertex] (v3) at (3.4, 0){ $d$};

        \draw (v0) edge[->, very thick, bend left = 55] node[anchor = south] { $0,\infty$}  (v1);
        \draw (v0) edge[->, very thick] node[anchor = south] { $1,\infty$}  (v1);
        \draw (v0) edge[->, very thick, bend right = 55] node[anchor = north] { $2,\infty$}  (v1);

        \draw (v1) edge[->, very thick] node[anchor = south] { $0,0$}  (v2);

        \draw (v2) edge[->, very thick, bend left = 55] node[anchor = south] { $0,\infty$}  (v3);
        \draw (v2) edge[->, very thick] node[anchor = south] { $1,\infty$}  (v3);
        \draw (v2) edge[->, very thick, bend right = 55] node[anchor = north] { $2,\infty$}  (v3);

      \node[anchor=north] at (1.7, 0.8){$T = 3, \Gamma = 2$};

      \draw[white] (0, -0.7) circle (1mm);
    \end{tikzpicture}
    \caption{\parbox{.9\textwidth}{Instance $I_3$ from Proposition \ref{prop:temp-rep:linear-gap}. Edge labels denote the travel times and possible delays, respectively.}}
    \label{fig:temp-rep:gap-instances:linear-gap}
\end{subfigure}
\end{center}
    \caption{Instances showing an optimality gap for temporally repeated flows.}
    \label{fig:temp-rep:gap-instances}
\end{figure}

\begin{restatable}{proposition}{proptempreplineargap}
 \label{prop:temp-rep:linear-gap}
 There are instances for robust maximum flow over time whose gap between an optimal temporally repeated flow and an optimal flow is $T$ and $\Gamma+1$.
\end{restatable}

\begin{proof}
 For fixed $r \in \ZZ_+$, we construct an instance $I_r$ with graph $G_r = (V,E_r)$ as follows (see Figure \ref{fig:temp-rep:gap-instances:linear-gap}).
 The vertex set always consists of four vertices $V = \{s,v_1,v_2,d\}$.
 The edge set consists of a designated edge $e^* = (v_1,v_2)$ and two sets of $r$ parallel edges $e^1_i = (s,v_1)$ and $e^2_i = (v_2,d), i = 0,\dots,r-1$.
 We set $\tau_{e^*} = \Delta_{e^*} = 0$, $\tau_{e^1_i} = \tau_{e^2_i} = i, \Delta_{e^1_i} = \Delta_{e^2_i} = \infty, \Gamma = r - 1, T = r$.
 All edges have unit capacity.
 Combination of the two claims below yields an optimality gap of $\Omega(r)$.

\omsClaim{There exists a solution to $I_r$ with objective function value $1$.}
 We define $r$ paths $P^i = \{e^1_i,e^*,e^2_{r-i-1}\}$ and dispatch a single unit of flow in the interval $[0,1)$ into each.
 The flow is clearly feasible and has a nominal objective function value of $r$.
 Moreover, any $r-1$ attacked edges can destroy at most $r-1$ units of flow, hence, yielding a flow value of at least $1$.

\omsClaim{An optimal temporally repeated flow can send at most $1 / r$ units of flow.}
We start by defining a temporally repeated flow with robust solution value $1/r$ and show the optimality of this flow. Let $x_i=1/r$ for all paths $P^i, 0\leq i \leq r-1$.

 We will argue that $x$ is optimal with robust function value $1/r$.
 Let us check the objective value contribution in detail.
 Each path $P^i, i= 0,\dots,r - 1$ contributes $$T - \tau_{e^1_i} - \tau_{e^2_{r-i-1}} = r - i - (r - i - 1) = 1$$
to the objective. Hence, sending a flow value $x_i = 1/r$ on all $r$ paths yields an objective function value of $1/r$ independent of which $\Gamma = r - 1$ edges between $s$ and $v_1$ are destroyed.
 Note that attacking different edges does not make any sense for this type of solution.

 Now, we show that any feasible temporally repeated flow different than $x$ yields a smaller objective function value. Let $x'$ be a different solution, i.e.\ there is some path $P^i$ that has flow rate $x'_i < 1/r$.
 We will show that this solution has objective value at most $x'_i$.

 More precisely, we will show that there is a scenario which leaves only the path $P^i$ intact.
 We will use the scenario $z \in \scenarios$ with $z_{e^1_j} = 1$ for all $0 \leq j < i$, $z_{e^2_j} = 1$ for all $0 \leq j < r - i - 1$ and zero otherwise.
 The total number of edges destroyed by $z$ is $r - 1$. In the following, we show that in this scenario only $P^i$ with $x'_i<1/r$ contributes to the objective function value.

 If any other path $P \neq P^i$ should survive the scenario, it must avoid all edges affected by $z$. $P$ can neither contain edges from $\{e_1^1,\dots,e_{i-1}^1\}$ nor from $\{e_1^2,\dots,e_{r-i-2}^2\}$. Let us consider the path length of such a path $P$. The only possibility to achieve a path length smaller than $T=r$ is to use edges $e_i^1$ and $e_{r-i}^2$, i.e.\ $P = P^i$. We conclude that the objective value is at most $x'_i < 1/r$, which proves the optimality of $x$ and finishes the proof of the claim.
 \qed
\end{proof}

\subsection{Asymptotic Optimality}
Although the gaps seem large, they appear only if the time horizon is
relatively short, when compared to the travel times.
An asymptotic bound shows that the optimality gap diminishes as the time horizon increases.

\begin{restatable}{proposition}{proptemprepasymptoticbound}
For instances of the maximum robust flow over time problem with $\Delta_e < \infty$ for all $e \in E$, temporally repeated flows tend to optimality for $T \rightarrow \infty$, if all other parameters are fixed.
 \label{prop:temp-rep:asymptotic-bound}
\end{restatable}

\begin{proof}
Let $(G,s,t,u,\Delta,\Gamma)$ with $\Delta_e < \infty$ for all $e \in E$ be an instance of robust max flow over time, with values that do not depend on the time horizon $T$.
If we denote the optimal objective function value of a temporally repeated flow by $f_{\text{TR}}^{\text{OPT}}(T)$ and the value of a general flow by $f^{\text{OPT}}(T)$ for time horizon $T$,
then for every $\epsilon > 0$ we show that there exists a $T' \in \RR$ such that  $f^{\text{OPT}}(T)/f_{\text{TR}}^{\text{OPT}}(T) \leq 1 + \epsilon$ for all $T \geq T'$.
 Since $\Delta$ is supposed to be a constant, we can assume that the choice of any scenario destroys at most a value of $\lambda^* = \Gamma \max_{e \in E} \Delta_e u_e$.
 For sufficiently large $T$, any optimal nominal temporally repeated flow will send a flow value of at least $F^*(T) - \lambda^*$, where $F^*(T)$ is the nominal optimal value for time horizon $T$.
 Moreover, as temporally repeated flows are optimal in the nominal case, we can deduce $f^{\text{OPT}} \leq F^*(T)$.
 Thus, $$ \frac{f^{\text{OPT}}(T)}{f_{\text{TR}}^{\text{OPT}}(T)} \leq \frac{F^*(T)}{F^*(T) - \lambda^*},$$ which tends to one as $T$ tends to infinity. \qed
\end{proof}

\subsection{Upper Bounds}
In the remainder of this section, we prove an upper bound on the gap between the objective function values of optimal general solutions and optimal temporally repeated flows.
This gap depends on some graph parameter $k$ introduced below.
Note that flow sent along a particular path $P\in \mathcal{P}$ only has a chance to reach the destination if each edge $e\in P$ is reached by the flow within interval
$I_{e,P}:=[\tau^{<e}(P), T - \tau^{\geq e}(P)]$, where $\tau^{<e}(P) = \sum_{e'\in E : e' <_P e}\tau(e')$ is the time required for a flow particle on path $P$ to reach edge $e$.
We call an instance of maximum flow over time \emph{$k$-coverable} if for each edge $e\in E$ it is possible to select $k$ points in time to cover all intervals $\{I_{e,P}\}_{P\in \mathcal{P}: e\in P}$.
The same definition holds for the robust counterpart. Note that this definition only depends on the graph $G$, the vertices $s,d$, the travel times $\tau$ and the time horizon $T$.
Denote these $k$ points in time by $t^e_1, \ldots t^e_k$.
We also call these points \emph{witnesses} of edge $e$.
That is, we define for each edge $e\in E$ the interval graph $H_e$ whose vertex set corresponds to the intervals $\{I_{e,P}\}_{P\in \mathcal{P}: e\in P}$ two of which are linked by an edge if and only if the associated intervals overlap.
Then the instance is $k$-coverable if and only if the vertices of each of the interval graphs $H_e$ can be covered by not more than $k$ cliques.
Since the clique-covering-number equals the maximum cardinality of a stable set by Dilworth's Theorem \cite{dilworth1950decomposition}, we define

\begin{definition}
An instance of maximum robust flow over time is \emph{$k$-coverable} if and only if the maximum size of a stable set in all of the associated interval graphs $H_e, e\in E,$ is at most $k$.
\end{definition}

For example, an instance satisfying $\max_{P \in \mathcal{P}} \tau^{<e}(P) + \max_{P \in \mathcal{P}} \tau^{\geq e}(P) \leq T$ for all $e \in E$ is $1$-coverable:
for an edge $e$, select $t^e_1 \in [\max_{P \in \mathcal{P}} \tau^{<e}(P), T-\max_{P \in \mathcal{P}} \tau^{\geq e}(P)]$.
This inequality is fulfilled in directed acyclic graphs whose longest path length does not exceed the time horizon, i.e.\ where $\max_{P \in \mathcal{P}} \tau(P) \leq T$. Note that maximum robust flow over time instances on a directed acyclic graph with the $T$-bounded path length property, i.e.\ with $\max_{P \in \mathcal{P}, z \in \scenarios} \tau(P)+\Delta_z(P) \leq T$, are also 1-coverable.

More vividly, let us consider the graphs in Figure \ref{fig:temp-rep:gap-instances}.
The graph in (a) has $k=1$, as its longest path length does not exceed the time horizon.
The graph in (b) has $k=2$, as can be seen by considering the edge $e^* = (v_1,v_2)$ and the following three paths.
Let $P_1$ be the path with the top edge from $s$ to $v_1$ and the bottom edge from $v_2$ to $d$.
It has interval $I_{e^*,P_1} = [0,1]$.
$P_2$ with both middle edges has $I_{e^*,P_2} = [1,2]$.
Finally, $P_3$ with the first bottom edge and the last top edge has $I_{e^*,P_3} = [2,3]$.
A minimum covering of these intervals can be achieved by choosing $t^{e^*}_1 = 1, t^{e^*}_2 = 2$.
Hence, $k$ is at least two.
One can easily check that the remaining intervals do not change this.
Moreover, one can observe that all remaining edges are 1-coverable.

It is shown in \cite{gupta1982efficient} that a stable set of maximum cardinality in an interval graph, in general, can be found by a simple greedy approach: \emph{sweep} from right to left through the whole domain, in our case $[0,T]$, and,
iteratively, select the interval with rightmost left endpoint. Remove this interval together with all intersecting intervals from the list, until no intervals are remaining.
In this work, we do not discuss the complexity status of computing $k$ in detail. We only note that the greedy approach described above is, in general, certainly not strongly polynomial as it depends on the time horizon $T$ and on the number of $s$-$d$ paths in the given graph.

Next, we prove upper bounds on the optimality gap of temporally repeated flows.
Since the proofs are rather technical, we distinguish between the case $\Delta_e \in \{0,\infty\}$ for all $e \in E$ and the case $\Delta_e \in \ZZ_+$.
The former turns out to be simpler and is covered in Theorem \ref{thm:temp-rep:k-log-T-gap}, the latter is covered in Theorem \ref{thm:temp-rep:delta-arbitrary} and builds upon the former.

\begin{theorem}
  \label{thm:temp-rep:k-log-T-gap}
  Let a $k$-coverable instance with $\Delta_e \in \{0,\infty\}$ for all $e \in E$ be given. Then, an optimal temporally repeated solution is an $O(k \log T)$-approximation for the robust maximum flow over time problem.
\end{theorem}

Additionally, we derive an upper bound for general instances.
Let $\bar{\Delta}_z(P)$ denote the effective amount of flow cut off by scenario $z \in \scenarios$ on path $P$, that is, $\bar{\Delta}_z(P) = \min \{ \Delta_z(P), T - \tau(P) \}$.
Then we define $$\eta = \max_{P \in \paths,z : \bar{\Delta}_z(P) < T - \tau(P)} \frac{T - \tau(P)}{T - \tau(P) - \bar{\Delta}_z(P)},$$ or $\eta = 1$, if no such path exists.
The value can be interpreted as follows.
If $\eta = 1$, the total contribution of flow on a path $P$ in a temporally repeated solution is either $x_P(T - \tau(P))$, or zero, depending on whether or not the path is attacked by the worst-case scenario.
If $\eta$ is large, the ratio between the contribution of path $P$, if $P$ is not attacked, and the non-zero contribution for some scenarios might be as large as $\eta$.
We will see that this ratio makes it harder to estimate the loss in the proof of Theorem \ref{thm:temp-rep:delta-arbitrary}.

\begin{theorem}
 \label{thm:temp-rep:delta-arbitrary}
Let a $k$-coverable instance with $\Delta \in \ZZ_+$ be given. Then,
an optimal temporally repeated solution is an $O(\eta k \log
T)$-approximation for the robust maximum flow over time problem.
\end{theorem}

The proof strategy of both theorems is the same and can be viewed as a dual fitting approach.
We present primal-dual pairs (\ref{LP:TR-P}), (\ref{LP:TR-D}) modeling robust temporally repeated flows and (\ref{LP:EX-P}), (\ref{LP:EX-D}) modeling general solutions to the robust flow over time problem.
It is clear by strong duality that $opt(\ref{LP:TR-P}) = opt(\ref{LP:TR-D})$ and $opt(\ref{LP:EX-P}) = opt(\ref{LP:EX-D})$.
Hence, in order to prove a bound on the optimality gap, it suffices to show $opt(\ref{LP:EX-D}) \leq \alpha opt(\ref{LP:TR-D})$, where $\alpha \geq 1$ is the upper bound on the optimality gap from Theorem \ref{thm:temp-rep:k-log-T-gap} and \ref{thm:temp-rep:delta-arbitrary}, respectively.
With this factor, we conclude
\begin{align}
 opt(\ref{LP:EX-P}) = opt(\ref{LP:EX-D}) \leq \alpha opt(\ref{LP:TR-D}) = \alpha opt(\ref{LP:TR-P})
 \Leftrightarrow \frac{f^{\text{OPT}}}{f^{\text{OPT}}_{\text{TR}}} = \frac{opt(\ref{LP:EX-P})}{opt(\ref{LP:TR-P})} \leq \alpha. \label{thm:temp-rep:bound-strategy}
\end{align}
We bound the factor $\alpha$ via a geometric box interpretation of solutions in the dual problems.
From an optimal solution of (\ref{LP:TR-D}), we construct a feasible solution of (\ref{LP:EX-D}), guaranteeing that the objective function values differ by at most the factor of $\alpha$.

In the remainder, we prove Theorems \ref{thm:temp-rep:k-log-T-gap} and \ref{thm:temp-rep:delta-arbitrary}.
We start by introducing the LP models. (\ref{LP:TR-P}) is a model for temporally repeated flows with corresponding dual (\ref{LP:TR-D}).
Without the constraints for scenarios $z$,  (\ref{LP:TR-P}) models (nominal) temporally repeated flows.
The variable $\lambda$ corresponds to the amount of flow that is lost due to the actions of the adversary.
 \begin{align*}
  \max_{x \in \RR_+^{|\mathcal{P}|},\ \lambda \in \RR_+}\quad & \sum_{P \in \mathcal{P}} (T - \tau(P) )x_P - \lambda \tag{P} \label{LP:TR-P} \\
  \text{s.t.} \quad & \sum_{\substack{P \in \mathcal{P} : \\ e \in P}} x_P \leq u_e & \quad && \forall e \in E \\
  & \sum_{\substack{P \in \mathcal{P} : \\ z \cap P \neq \emptyset}} \bar{\Delta}_z(P) x_P - \lambda \leq 0 & \quad && \forall z \in \scenarios
 \end{align*}
The corresponding dual with variables $\alpha_e$ and $\beta_z$ is:
 \begin{align*}
  \min_{\alpha \in \RR_+^{|E|},\ \beta \in \RR_+^{|\scenarios|}}\quad & \sum_{e \in E} u_e \alpha_e \tag{D} \label{LP:TR-D} \\
  \text{s.t.} \quad & \sum_{z \in \scenarios} \beta_z \leq 1  \\
  & \sum_{e \in P} \alpha_e + \sum_{\substack{z \in \scenarios : \\ z \cap P \neq \emptyset}} \bar{\Delta}_z(P) \beta_z \geq T - \tau(P) & \quad && \forall P \in \paths
 \end{align*}

Now, we present a model (\ref{LP:EX-P})  for general robust solutions.
We introduce variables $x_P^i$ which correspond to the flow rate sent into path $P$ in the dispatch interval $[i, i+1)$. Here, we implicitly use the fact that we can restrict to
integer dispatch intervals, as shown in Proposition~\ref{prop:modelling-techniques:piecewise-constant}. As before, $\lambda$ is the loss incurred after the adversary acts.
For ease of notation, we assume that all variables with a negative index in (\ref{LP:EX-P}) are omitted, e.g.\ a variable $x_P^{-5}$ is assumed to be excluded from the model.
This can also be seen as having variables $x_P^i$ for all $i \in \ZZ$ and forcing $x_P^i = 0$ for all $i < 0$ or $i \geq T - \tau(P)$.
We use the notation $\Delta^{<e}_z(P) = \sum_{e' \in P : e' <_P e} \Delta_{e'}= z_{e'}$ to denote the delay on path $P$ induced by scenario $z$ until edge $e$.
\begin{align*}
 \max_{x,\ \lambda \geq 0}\quad & \sum_{P \in \mathcal{P}} \sum_{0 \leq i < T - \tau(P)} x_P^i - \lambda \tag{P'} \label{LP:EX-P} \\
 \text{s.t.} \quad & \sum_{\substack{P \in \mathcal{P} :\\ e \in P}} x_P^{t - \tau^{<e}(P) - \Delta^{<e}_z(P) } \leq u_e \quad && \forall e \in E, \forall 0 \leq t < T, \forall z \in \scenarios \\
 & \sum_{P \in \mathcal{P}} \sum_{i \geq T - \tau(P) - \Delta_z(P)} x_P^i - \lambda \leq 0 \quad && \forall z \in \scenarios
\end{align*}
The corresponding dual with variables $\alpha_e^t(z)$ and $\beta_z$ is:
\begin{align*}
 \min_{\alpha,\ \beta \geq 0}\quad & \sum_{e \in E} u_e \sum_{ 0 \leq t < T} \sum_{z \in \scenarios} \alpha_e^t(z) \tag{D'} \label{LP:EX-D} \\
 \text{s.t.} \quad & \sum_{z \in \scenarios} \beta_z \leq 1 \\
 & \sum_{e \in P} \sum_{z \in \scenarios} \alpha_e^{ i + \tau^{<e}(P) + \Delta^{<e}_z(P) }(z) + \sum_{z \in \scenariosDualPath{i}{P} } \beta_z \geq 1 \quad && \forall P \in \paths, \forall 0 \leq i < T - \tau(P).
\end{align*}
By $\scenariosDualPath{i}{P}$ we denote for a path $P$ and time $i$ the set of scenarios for which flow on path $P$ sent into the network at time $i$ will not reach the destination, that is, $\scenariosDualPath{i}{P} = \{ z \in \scenarios : i \geq T - \tau(P) - \Delta_z(P) \}$.
In other words, these are the scenarios for path $P$ which prevent flow sent at time $i$ from contributing to the objective.

\sectionheadline{Graphical Interpretation of (\ref{LP:TR-D})}

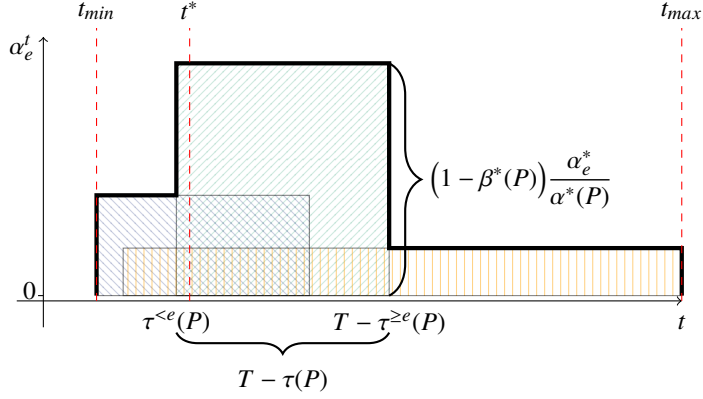
\begin{figure}[tb]
\begin{center}
\begin{tikzpicture}[scale=0.7, every node/.style={scale=1.2}]
  \draw[fill=white,draw=white,use as bounding box] (0,-1) rectangle (15.5,6.5);

  \begin{scope}
    \draw[->] (0.5,1) -- (13,1);
    \draw[->] (1,0.5) -- (1,6);

    \node at (0.74,1.2) {$0$};
    \draw[] (0.9,1.1) -- (1,1.1);

    \node at (0.6,5.8) {$\alpha^t_{e}$};
    \node at (13,0.6) {$t$};
  \end{scope}

  \begin{scope}

    \transparent{0.5}
    \draw[pattern=north east lines,pattern color=darkgreen] (3.5,1.1) rectangle (7.5,5.5);
    \draw[pattern=north west lines,pattern color=darkblue] (2.0,1.1) rectangle (6.0,3.0);
    \draw[pattern=vertical lines,pattern color=darkorange] (2.5,1.1) rectangle (13.0,2.0);
    \transparent{1.0}

    \draw[-,ultra thick, draw=black] (2.0,1.1) -- (2.0,3.0) -- (3.5,3.0) -- (3.5,5.5) -- (7.5,5.5) -- (7.5,2.0) -- (13.0,2.0) -- (13.0,1.1);

    \draw[-,dashed,draw=red] (2.0,1.0) -- (2.0,6.5);
    \draw[-,dashed,draw=red] (13.0,1.0) -- (13.0,6.5);
    \draw[-,dashed,draw=red] (3.75,1.0) -- (3.75,6.5);

    \node[fill=white] at (2.0,6.5) {$t_{min}$};
    \node[fill=white] at (13.0,6.5) {$t_{max}$};
    \node[fill=white] at (3.75,6.5) {$t^*$};

    \node at (3.5,0.6) { $\tau^{<e}(P)$};
    \node at (7.5,0.6) { $T - \tau^{\geq e}(P)$};
    \draw[thick,black,decorate,decoration={brace,amplitude=8pt}] (7.5,0.4) -- (3.5,0.4) node[midway, below,yshift=-0.3cm]{$T - \tau(P)$};

    \draw[thick,black,decorate,decoration={brace,amplitude=12pt}] (7.5,5.5) -- (7.5,1.1) node[midway, right,xshift=0.3cm]{ $\begin{aligned} \Bigl( 1 - \beta^*(P) \Bigr) \frac{\alpha^*_e}{\alpha^*(P)} \end{aligned}$ };
  \end{scope}
\end{tikzpicture}
\end{center}
  \caption{Illustration of the box model with three boxes. The height profile is a solution constructed for the proof of Theorem \ref{thm:temp-rep:k-log-T-gap}.}
  \label{fig:temp-rep:dual-box-model}
\end{figure}
 The remaining proofs in this section rely on the following graphical interpretation of the duals (\ref{LP:TR-D}) and (\ref{LP:EX-D}).
 We consider an optimal solution $(\alpha^*, \beta^*) \in (\ref{LP:TR-D})$ and interpret the dual constraints and variables of (\ref{LP:TR-D}) as boxes in a two-dimensional space $[0,T] \times [0,1]$ as follows (see Figure \ref{fig:temp-rep:dual-box-model}):
 Let us consider the dual constraint for a path $P \in \paths$ and assume $\bar{\Delta}_z(P) \in \{0, T - \tau(P)\}$ for all $z \in \scenarios$, that is, $P$ only contains edges that have a very large value for $\Delta$, or zero.
 The general case is discussed in the proof of Theorem \ref{thm:temp-rep:delta-arbitrary}.

 We can regard the dual constraint of such a path $P$ as a set of boxes
 \[ B^e_P = \left[ \tau^{<e}(P), T - \tau^{\geq e}(P) \right] \times \left[ 0, \Bigl( 1 - \beta^*(P) \Bigr) \frac{\alpha_{e}^*}{\alpha^*(P)} \right], e \in P, \]
 where $\alpha^*(P) = \sum_{e \in P} \alpha^*_e$ and $\beta^*(P) = \sum_{z \in \scenarios : z \cap P \neq \emptyset} \beta^*_z$.
 Each box $B^e_P$ starts at the earliest possible arrival of flow at edge $e$ traveling along path $P$ and ends at the latest reasonable departure time from edge $e$.
 As soon as we consider the total area covered by the boxes of a path $P$, the connection between the boxes and their dual constraint becomes evident.
  \begin{align*}
  \sum_{e \in P} \text{vol}(B^e_P) = \sum_{e \in P} \left(T - \tau(P)\right)\Bigl(1 -  \beta^*(P)\Bigr)\frac{\alpha^*_e}{\alpha^*(P)}
  = T - \tau(P) - \sum_{\substack{z \in \scenarios : \\ z \cap P \neq \emptyset}} \beta^*_z \bar{\Delta}_z(P)
 \end{align*}
 Here we used that $\bar{\Delta}_z(P) = T -\tau(P)$ for $z \cap P \neq \emptyset$ as $\Delta_e \in \{0, \infty\}$.
 The sum of volumes of all boxes of a path is equal to the corresponding right hand side value of its dual constraint in (\ref{LP:TR-D}) (assuming all terms involving $\beta^*$ are moved to the right hand side).
 Furthermore, the volume of $B^e_P$ is a lower bound for the value of $\alpha^*_e$ for all edges $e \in P$.
 This can be seen by considering the dual constraint for path $P$ in (\ref{LP:TR-D}), multiplying it by $\alpha^*_e$ and rearranging the terms to have $vol(B^e_P)$ as the right hand side value.
 \begin{align}
  \alpha^*(P) = \sum_{e' \in P} \alpha^*_{e'} \geq \left( 1- \beta^*(P) \right) \left(T - \tau(P) \right) \notag \\
  \Leftrightarrow \alpha^*_e \alpha^*(P) \geq \alpha^*_e \left( 1- \beta^*(P) \right) \left( T - \tau(P) \right) \notag \\
  \Leftrightarrow \alpha^*_e \geq \frac{\alpha^*_e}{\alpha^*(P)} \left(1- \beta^*(P) \right) \left(T - \tau(P) \right) = vol(B^e_P). \label{bounds:alpha-volume-bound}
 \end{align}
 In the following, we will use the \emph{area covered} by boxes in order to construct a feasible solution for (\ref{LP:EX-D}) and use the \emph{volume} of the boxes in order to estimate the total cost of the solution.

 In a $k$-coverable instance, every box $B^e_P$ intersects at least one of the lines $t^e_i \times [0, 1]$.
 Recall that $t^e_i$ are the witnesses of edge $e$, that is, the points used in order to cover the interval graph in the definition of $k$-coverability.
 In the interval graph $H_e$, we had intervals $I_{e,P}:=[\tau^{<e}(P), T - \tau^{\geq e}(P)]$ for path $P$.
 Note that this is exactly the shadow of box $B^e_P$ on the horizontal $t$-axis.

 For the remainder of this part, we will restrict all arguments to boxes which intersect a specific line $t^* = t^e_i \times [0, 1]$.
 That is, we will consider only boxes $\{B^e_P : \tau^{<e}(P) \leq t^* \leq T - \tau^{\geq e}(P) \}$.
 Afterwards, in the proofs of the theorems, we will do a summation over all lines $t^e_i$.
 A box may intersect multiple lines - and thus be considered multiple times - but this does not harm the approximation guarantee.
 In particular, we assume that no box starts after, respectively ends before, $t^*$.
 By $t_{min}$ and $t_{max}$ we denote the earliest and latest points on $[0,T]$ covered by any of the considered boxes (see Figure \ref{fig:temp-rep:dual-box-model}).
 That is, $t_{min} = \min \{ \tau^{<e}(P) : \tau^{<e}(P) \leq t^* \leq T - \tau^{\geq e}(P) \}$ and $t_{max}$ is defined analogously.

\sectionheadline{Construction of a feasible solution for (\ref{LP:EX-D})}
 Let us examine feasible values for the dual variables $\alpha_e^t(z)$ in (\ref{LP:EX-D}) which can be constructed from the graphical interpretation of $(\alpha^*,\beta^*)$.

 First, let us see how $\alpha_e^t(z)$ can be set in order to become feasible for the constraints induced by a single path $P$.
 We set $\alpha_e^t(z) = 0$ for all $z \neq \emptyset$.
 Setting these variables to a positive value would be hard to analyze.
 But in order to become feasible, we set $\alpha_e^t(\emptyset)$ to the height of box $B^e_P$ in each slice $(t,t+1)$.
 This results in $\alpha_e^t(\emptyset)$ being the height of $B^e_P$, if $\tau^{<e}(P) \leq t < T - \tau^{\geq e}(P)$, or zero otherwise.
 Then
 \begin{align*}
  \sum_{e \in P} \sum_{z \in \scenarios} \alpha_e^{i + \tau^{<e}(P) + \Delta_z^{<e}(P)}(z)
  = \sum_{e \in P} \left( 1 - \sum_{z \in \scenariosDualPath{i}{P}} \beta^*_z \right)\frac{\alpha_e^*}{\alpha^*(P)}
  = 1 - \sum_{\substack{z \in \scenariosDualPath{i}{P}}} \beta^*_z.
 \end{align*}
 That means, if we had only a constraint for a single path, the solution would be feasible.
 But we have many paths and the height of their boxes $B^e_P$ differs.

 Hence, if we set $\alpha_e^t(\emptyset)$ to the maximal height of all boxes $B^e_P$ intersecting slice $(t,t+1)$, we obtain a feasible solution $(\alpha,\beta^*) \in (\ref{LP:EX-D})$.
 Formally, the solution can be described as
 \[ \alpha_e^t(\emptyset) = \max_{P \in \mathcal{P}_e^t} \left\{ \left( 1 - \sum_{z \in \scenariosDualPath{t}{P} } \beta^*_z \right)\frac{\alpha_e^*}{\alpha^*(P)} \right\}, \]
 where $\mathcal{P}_{e}^t = \{P : e \in P \wedge \tau^{<e}(P) \leq t < T - \tau^{\geq e}(P) \}$ describes the set of paths which can utilize edge $e$ at time $t$.
 And $\alpha_e^t(z) = 0$ for all other variables.
 The height profile in Figure \ref{fig:temp-rep:dual-box-model} depicts this solution.

 The constructed solution is feasible for (\ref{LP:EX-D}).
 We will now analyze its total cost.
 Intuitively, we argue as follows:

 The sum of dual variables, $\sum_{0 \leq t < T} \alpha_{e}^t(\emptyset)$ for any fixed edge $e$ is equal to the area covered by the union of all boxes $B^e_P$ induced by paths $P \in \mathcal{P}$ such that $e \in P$.
 Thus, in order to prove the $O(k \log T)$-approximation factor for Theorem~\ref{thm:temp-rep:k-log-T-gap}, we will show that, for any fixed edge $e \in E$, the total area covered by the union of all boxes is bounded by $O(k \log T)$ times the area of the largest box.
 Recall that (\ref{bounds:alpha-volume-bound}) implies that the value of $\alpha^*_e$ in the temporally repeated solution is lower bounded by the area of the largest box.
 With this, we can estimate the total solution cost of $(\alpha,\beta^*)$ in terms of $(\alpha^*,\beta^*)$.

 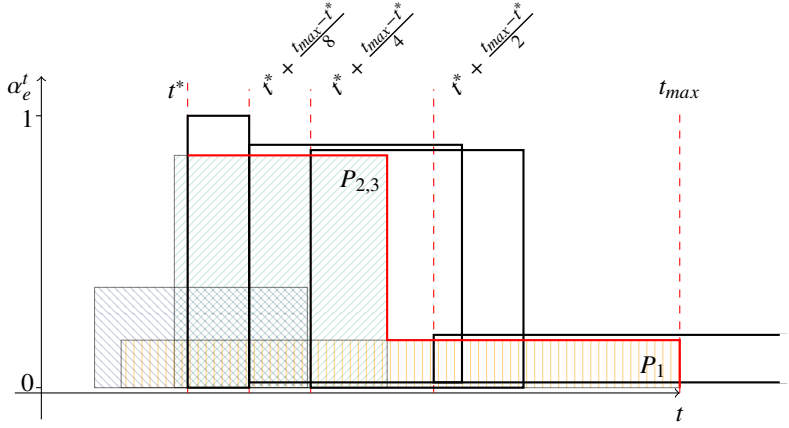
\begin{figure}[tb]
\begin{center}
\begin{tikzpicture}[scale=0.7, every node/.style={scale=1.2}]
  \draw[fill=white,draw=none,use as bounding box] (0,0.25) rectangle (15.5,8.5);

  \begin{scope}
    \draw[->] (0.5,1) -- (13,1);
    \draw[->] (1,0.5) -- (1,7);

    \node at (0.74,1.2) {$0$};
    \draw[] (0.9,1.1) -- (1,1.1);

    \node at (0.74,6.2) {$1$};
    \draw[] (0.9,6.25) -- (1,6.25);

    \node at (0.6,6.8) {$\alpha^t_{e}$};
    \node at (13,0.6) {$t$};
  \end{scope}

  \begin{scope}
    \transparent{0.5}
    \draw[pattern=north east lines,pattern color=darkgreen] (3.5,1.1) rectangle (7.5,5.5);
    \draw[pattern=north west lines,pattern color=darkblue] (2.0,1.1) rectangle (6.0,3.0);
    \draw[pattern=vertical lines,pattern color=darkorange] (2.5,1.1) rectangle (13.0,2.0);
    \transparent{1.0}

    \draw[-,dashed,draw=red] (13,1.0) -- (13,7);
    \draw[-,dashed,draw=red] (3.75,1.0) -- (3.75,7);
    \node[fill=white] at (13,6.75) {$t_{max}$};
    \node[fill=none] at (3.55,6.75) {$t^*$};

    \draw[-,dashed,draw=red] (8.375,1.0) -- (8.375,6.75);
    \node[fill=none,rotate = 45] at (9.5,7.5) {$t^* + \frac{t_{max} - t^*}{2}$};
    \node at (12.5,1.5) {$P_1$};
    \draw[draw=black,thick] (8.375,1.2) rectangle (15,2.1);
    \draw[draw=white,fill=white,thick] (14.9,1.2) rectangle (15.1,2.1);

    \draw[-,dashed,draw=red] (6.0625,1.0) -- (6.0625,6.75);
    \node[fill=none, rotate = 45] at (7.25,7.5) {$t^* + \frac{t_{max} - t^*}{4}$};
    \node at (7.0,5.0) {$P_{2,3}$};
    \draw[draw=black,thick] (6.0625,1.1) rectangle (10.0625,5.6);

    \draw[-,dashed,draw=red] (4.90625,1.0) -- (4.90625,6.75);
    \node[fill=none, rotate = 45] at (5.95,7.5) {$t^* + \frac{t_{max} - t^*}{8}$};

    \draw[draw=black,thick] (4.90625,1.2) rectangle (8.90625,5.7);

    \draw[draw=black,thick] (3.75,1.1) rectangle (4.90625,6.25);

    \draw[draw=red,thick] (3.75,5.5) -- (7.5,5.5) -- (7.5,2.0) -- (13.0,2.0) -- (13.0,1.1);
  \end{scope}
\end{tikzpicture}
\end{center}
  \caption{Illustration of the proof of Theorem \ref{thm:temp-rep:k-log-T-gap}.
  The interval $[t^*,t_{max}]$ contains three boxes (orange,blue,green).
  The total area is covered by copies of $P_i, i \geq 1$, the copies are black (slightly moved up/down to distinguish different copies).
  The box used in iteration $2$ and $3$ is the same, moved to a different offset.
  We also assume that $\frac{t_{max} - t^*}{8} \leq \alpha^*_e$, hence, the final box is just the slice of height $[0,1]$.
  The height profile of the solution is depicted in red.}
  \label{fig:temp-rep:dual-covering}
\end{figure}

\begin{proof}[Proof of Theorem \ref{thm:temp-rep:k-log-T-gap}]
In order to prove the theorem, we will show $opt(\ref{LP:EX-D}) \leq O(k \log T) opt(\ref{LP:TR-D})$, which is sufficient in combination with (\ref{thm:temp-rep:bound-strategy}).
Therefore, we will use the argumentation from the preceding paragraphs, in particular, we will use the constructed solution $\alpha_e^t(z)$.
We will show
\begin{align}
 \sum_{ 0 \leq t < T} \sum_{z \in \scenarios} \alpha_e^t(z)
 = vol \left(\bigcup_{\substack{P \in \mathcal{P} :\\ e \in P}} B^e_P \right)
 \leq \sum_{i=1}^{k} vol \left(\bigcup_{\substack{P \in \mathcal{P}_e^{t_i}}} B^e_P \right)
 \leq O(k \log T) \alpha^*_e \quad \forall e \in E. \label{proof:temp-rep:l-log-T-gap:pointwise}
\end{align}
Recall that $(\alpha^*, \beta^*)$ was an optimum solution to (\ref{LP:TR-D}).
Thus, result (\ref{proof:temp-rep:l-log-T-gap:pointwise}) then will prove
$$ opt(\ref{LP:EX-D}) \leq \sum_{e \in E} u_e \sum_{ 0 \leq t < T} \sum_{z \in \scenarios} \alpha_e^t(z) \leq O(k \log T) \sum_{e \in E} u_e \alpha^*_e = O(k \log T) opt(\ref{LP:TR-D}). $$

 Without loss of generality, we can assume $\sum_{z \in \scenarios}\beta^*_z = 1$.
 For the remainder of the proof, we will consider the graphical interpretation of the duals from the preceding paragraphs and fix an edge $e \in E$ and witness at $t^* = t^e_i$.
 For ease of notation, we use $\mathcal{P}$ to denote the set $\mathcal{P}_e^{t^*}$ of all paths whose boxes intersect $t^*$.
 If a path intersects multiple witnesses, we consider it multiple times.

 We will show that the total volume of the union of the boxes intersecting $t^*$ can be bounded by $O(\log T)\alpha_e^*$.
 Summation over $t_i^e$ with $i \le k$ yields the result.

 Since every box $B^e_P$ for path $P \in \mathcal{P}$ intersects the point $t^*$, it is easy to see that the sets of paths $\mathcal{P}^\ell = \{P \in \mathcal{P} : T - \tau^{\geq e}(P) \geq t^* + (t_{max} - t^*) 2^{-\ell} \}$ are monotonically increasing, that is, $\mathcal{P}^i \subseteq \mathcal{P}^{i+1}$.
 The set $\mathcal{P}^\ell$ contains all rectangles covering some area to the right of the line $t^* + (t_{max} - t^*) 2^{-\ell}$.

 Now, let us consider the set $\mathcal{P}^1$ and pick the path $P_1 \in \mathcal{P}^1$ whose box $B_1$ is the highest among all candidates, that is,
 \[ P_1 = \arg\max_{P \in \mathcal{P}^1} \left( 1 - \sum_{z \in \scenariosDualPath{t^*}{P}} \beta^*_z \right)\frac{\alpha_e^*}{\alpha^*(P)}. \]
 $B_1$ has the largest height among all boxes which intersect the interval $[t^* + (t_{max} - t^*)\frac{1}{2}, t_{max}]$.
 Furthermore, it has a total width exceeding the width of the interval, as it contains both, $t^*$ and $t^* + (t_{max} - t^*)\frac{1}{2}$.
 Hence, we can take a copy of $B_1$, shift it such that it starts at $t^* + (t_{max} - t^*)/2$ and cover the total area of the interval $[t^* + (t_{max} - t^*)/2, t_{max}]$.
 An example is given in Figure \ref{fig:temp-rep:dual-covering}.

 Analogously, we can proceed with increasing $\ell$ in order to pick a path $P_\ell \in \mathcal{P}^\ell$ whose box $B_\ell$ covers the interval $[t^* + (t_{max} - t^*) 2^{-\ell}, t^* + (t_{max} - t^*) 2^{-\ell+1}]$.
 As soon as $t^* + (t_{max} - t^*) 2^{-\ell} \leq \alpha^*_e$, we can stop as the remaining area in the interval $[t^*, t^* + \alpha^*_e]$ is at most $\alpha^*_{e}$.
 The total area consumption of chosen boxes is bounded by $O(\log T) \alpha^*_e$ as the length of the remaining interval to be covered is divided by two in each step.

 In summary, we have just shown that the total area covered by boxes in the interval $[t^*, t_{max}]$ can be bounded.
 By symmetry, the same arguments also hold for the area covered by boxes in the interval $[t_{min}, t^*]$.
 Here, we shift boxes to the left instead of to the right.
 So the total area covered by boxes that intersect $t^*$ can be bounded by $2 \cdot O(\log T) \alpha^*_e$.
 Recall $\mathcal{P} = \cup_{1 \leq i \leq k} \{ P \in \mathcal{P} : P \cap t^e_i \neq \emptyset \}$, hence,
 $$vol \left(\bigcup_{\substack{P \in \mathcal{P} :\\ e \in P}} B^e_P \right)
 \leq \sum_{i=1}^{k} vol \left(\bigcup_{\substack{P \in \mathcal{P}_e^{t_i}}} B^e_P \right) \leq \sum_{i=1}^k 2 \cdot O(\log T) \alpha^*_e = O(k \log T) \alpha^*_e$$
 holds, which concludes the proof. \qed
\end{proof}

Since we restricted the preceding argumentation only to the case when $\Delta \in \{0,\infty\}$, we will now generalize it towards arbitrary delays in the proof of Theorem \ref{thm:temp-rep:delta-arbitrary}.

\begin{proof}[Proof of Theorem \ref{thm:temp-rep:delta-arbitrary}]
 We begin the proof by generalizing the box model to general scenarios.
 Let us fix a path $P$.
 Again, we want to construct geometric objects $B^e_P$ such that $\sum_{e \in P} vol(B^e_P)$ equals the right hand side value of the dual constraint of path $P$ in (\ref{LP:TR-D}) (assuming that all terms involving $\beta$ are moved to the right hand side).

 Therefore, let us start with the following boxes
 \[ \hat{B}^e_P = \left[ \tau^{<e}(P), T - \tau^{\geq e}(P) \right] \times \left[ 0, \frac{\alpha^*_e}{\alpha^*(P)} \right]. \]
 Let us also fix an edge $e \in P$ and order the scenarios $z_i, 1 \leq i \leq \ell$ with $\beta^*_{z_i} > 0$ such that $\bar{\Delta}_{z_i}(P) \geq \bar{\Delta}_{z_{i+1}}(P)$ for all $1 \leq i < \ell$.
 For each $1 \leq i \leq \ell$, we cut out a section of box $B^e_P$ of height $\beta^*_{z_i}$ starting at the top right boundary of the box.
 That is, for each scenario, define the following rectangles
 \[ A^e_{z_i} = \left[ T - \tau^{\geq e}(P) - \bar{\Delta}_{z_i}(P), T - \tau^{\geq e}(P) \right] \times \left[ \left( 1 - \sum_{j=1}^i \beta^*_{z_j} \right) \frac{\alpha^*_e}{\alpha^*(P)}, \left( 1 - \sum_{j=1}^{i-1} \beta^*_{z_j} \right) \frac{\alpha^*_e}{\alpha^*(P)} \right] \]
 and cut these out of $\hat{B}^e_P$.
 So we end up with a polygon $B^e_P = \hat{B^e_P} \setminus \left( \bigcup_{i=1}^\ell A^e_{z_i} \right)$ (see the box shaped as a staircase in Figure \ref{fig:temp-rep:dual-box-model-delta-arb}).
 Also note that the boxes $A$ do not overlap - except for their boundaries - and are always contained in $\hat{B^e_P}$.
 Hence, it is simple to compute the total volume of the polygons of any path.
 \begin{align*}
\sum_{e \in P} vol(B^e_P)
&= \sum_{e \in P} \left( \frac{\alpha^*_e}{\alpha^*(P)} \left( T - \tau(P) \right) - \sum_{i=1}^{\ell} vol(A_{z_i}) \right) \\
&= \sum_{e \in P} \left( \frac{\alpha^*_e}{\alpha^*(P)} \left( T - \tau(P) \right) - \sum_{i=1}^{\ell} \bar{\Delta}_{z_i}(P) \beta^*_{z_i} \frac{\alpha^*_e}{\alpha^*(P)} \right) \\
&= T - \tau(P) - \sum_{\substack{z \in \scenarios : \\ z \cap P \neq \emptyset}} \bar{\Delta}_z(P) \beta^*_z.
\end{align*}
Analogously to (\ref{bounds:alpha-volume-bound}) we can deduce that $\alpha^*_e \geq vol(B^e_P)$ holds for all $P$ and $e \in P$, in other words, the volume of polygons is again a lower bound on the variables of the dual temporally repeated solution.

\begin{figure}[tb]
\centering
\resizebox{0.7\textwidth}{!}{
\begin{tikzpicture}[scale=1, every node/.style={scale=2}]
  \draw[fill=white,draw=none,use as bounding box] (0,0.25) rectangle (12.5,6.85);

  \begin{scope}
    \draw[->] (0.5,1) -- (11,1);
    \draw[->] (1,0.5) -- (1,6);

    \node at (0.74,1.2) {$0$};
    \draw[] (0.9,1.1) -- (1,1.1);

    \node at (0.6,5.8) {$\alpha^t_{e}$};
    \node at (11,0.6) {$t$};
  \end{scope}

  \begin{scope}

    \draw[draw=darkblue,ultra thick] (2,1.1) -- (2,5) -- (4,5) -- (4,4) -- (5,4) -- (5,2) -- (7,2) -- (7,1.1) -- (2,1.1);
    \draw[draw=red,dashed, ultra thick] (1.95,1.05) rectangle (7.05,5.05);

    \node at (2.5,5.5) {$A_{z_1}$};
    \node at (4.5,4.5) {$A_{z_2}$};
    \node at (6,3) {$A_{z_3}$};

    \transparent{0.4}
    \draw[pattern=north east lines,pattern color=red] (5,2) rectangle (7,4);
    \draw[pattern=north west lines,pattern color=darkblue] (4,4) rectangle (7,5);
    \draw[pattern=vertical lines,pattern color=darkorange] (2,5) rectangle (7,6);

    \transparent{1}

    \draw[thick,black,decorate,decoration={brace,amplitude=8pt}] (7,6) -- (7,5) node[midway, right,xshift=+0.25cm]{$\beta_{z_1}(P) \frac{\alpha^*_e}{\alpha(P)}$};
    \draw[thick,black,decorate,decoration={brace,amplitude=8pt}] (2,6) -- (7,6) node[midway, above,yshift=0cm]{$\bar{\Delta}_{z_1}(P)$};

    \draw[thick,black,decorate,decoration={brace,amplitude=8pt}] (7,5) -- (7,4) node[midway, right,xshift=+0.25cm]{$\beta_{z_2}(P) \frac{\alpha^*_e}{\alpha(P)}$};
    \draw[thick,black,decorate,decoration={brace,amplitude=8pt}] (4,5) -- (7,5) node[midway, above,yshift=0cm]{$\bar{\Delta}_{z_2}(P)$};

    \draw[thick,black,decorate,decoration={brace,amplitude=8pt}] (7,4) -- (7,2) node[midway, right,xshift=+0.25cm]{$\beta_{z_3}(P) \frac{\alpha^*_e}{\alpha(P)}$};
    \draw[thick,black,decorate,decoration={brace,amplitude=8pt}] (5,4) -- (7,4) node[midway, above,yshift=0cm]{$\bar{\Delta}_{z_3}(P)$};
  \end{scope}
\end{tikzpicture}}
  \caption{Illustration of the box model for arbitrary values of $\Delta$.
  The boundary of the polygon $B^e_P$ is drawn in solid blue.
  The corresponding rectangular box $\bar{B}^e_P$ from the proof of Theorem \ref{thm:temp-rep:delta-arbitrary} is drawn red, dashed.}
  \label{fig:temp-rep:dual-box-model-delta-arb}
\end{figure}
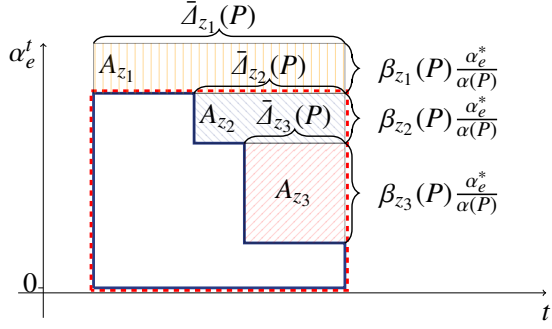

Note that the preceding \emph{construction of a feasible solution for (\ref{LP:EX-D})} used only variables $\alpha_e^t(\emptyset)$.
Hence, the solution constructed therein is not only feasible if $\Delta_e \in \{0,\infty\}$, but also for general $\Delta_e \in \ZZ_+$ by the same arguments.
Fix a path $P$ and set $\alpha_e^t(\emptyset)$ to the height of polygon $B^e_P$ in the slice $(t,t+1)$.
This time, note that the height of a polygon corresponding to a specific path $P$ and an edge $e \in P$ may differ depending on the actual slice.
Then the solution is feasible as the height of polygons in slice $(t,t+1)$ is, by construction, equal to $\frac{\alpha^*_e}{\alpha^*(P)} \left( 1 - \sum_{z \in \scenariosDualPath{t}{P}} \beta^*_z \right)$.

It remains to prove that the total cost of the constructed solution $\alpha$ is bounded by the total cost of the temporally repeated solution times at most a factor $O(\eta k \log T)$.
Unfortunately, the staircase structure of $B^e_P$ prevents us from covering the area as in the proof of Theorem \ref{thm:temp-rep:k-log-T-gap}.
If we take a copy and shift it to the left, the area to the left of $t^*$ will not necessarily be covered.

Therefore, we will use the following strategy.
In a first step, we will replace all polygons $B^e_P$ from the temporally repeated solution by rectangles $\bar{B}^e_P$.
This will cost at most a factor $\eta$.
The polygons will be replaced in such a way that the height profile with respect to $B^e_P$ is contained in the height profile with respect to $\bar{B}^e_P$.
Afterwards, we can apply exactly the same proof as in Theorem \ref{thm:temp-rep:k-log-T-gap}.
The proof will be concluded.

The replacement is done as follows.
Fix a path $P$ and $e \in P$ and consider the polygon $B^e_P$.
We want to replace it by a rectangle $\bar{B}^e_P$ that contains the polygon.
Therefore, let $h = \frac{\alpha^*_e}{\alpha^*(P)} \left( 1 - \sum_{z \in \scenarios : \bar{\Delta}_z(P) = T - \tau(P)} \beta^*_z \right)$ be the highest point of the polygon.
We will replace $B^e_P$ by the rectangle
$\bar{B}^e_P = \left[ \tau^{<e}(P), T - \tau^{\geq e}(P) \right] \times \left[ 0, h \right]$.
Clearly, it contains the polygon.
The factor we lose can be bounded as follows.
We use $\scenarios_1 = \{z \in \scenarios : \bar{\Delta}_z(P) = T - \tau(P) \}$ and $\scenarios_2 = \scenarios \setminus \scenarios_1$ to partition the scenarios.
The first set contains all scenarios that destroy the entire flow on path $P$, the latter contains scenarios that destroy only a smaller amount.
\begin{align*}
 \frac{vol(\bar{B}^e_P)}{vol(B^e_P)}
 &= \frac{\frac{\alpha^*_e}{\alpha^*(P)} \left( T - \tau(P) \right) \left( 1 - \sum_{z \in \scenarios_1} \beta^*_z \right)}{\frac{\alpha^*_e}{\alpha^*(P)} \left(T - \tau(P) - \sum_{z \in \scenarios} \beta^*_z \bar{\Delta}_z(P) \right)} \\
 &= \frac{\left( T - \tau(P) \right) \left( 1 - \sum_{z \in \scenarios_1} \beta^*_z \right)}{ T - \tau(P) - \sum_{z \in \scenarios_1} \beta^*_z \left(T - \tau(P) \right) - \sum_{z \in \scenarios_2} \beta^*_z \bar{\Delta}_z(P)} \\
 &= \frac{\left( T - \tau(P) \right) \left( 1 - \sum_{z \in \scenarios_1} \beta^*_z \right)}{ \left(T - \tau(P) \right) \left( 1 - \sum_{z \in \scenarios_1} \beta^*_z \right) - \sum_{z \in \scenarios_2} \beta^*_z \bar{\Delta}_z(P)} \\
 &\leq \frac{\left( T - \tau(P) \right) \left( 1 - \sum_{z \in \scenarios_1} \beta^*_z \right)}{ \left(T - \tau(P) \right) \left( 1 - \sum_{z \in \scenarios_1} \beta^*_z \right) - \max_{z \in \scenarios_2}\{ \bar{\Delta}_z(P)\} \sum_{z \in \scenarios_2} \beta^*_z} \\
 &= \frac{T - \tau(P)}{ T - \tau(P) - \max_{z \in \scenarios_2}\{ \bar{\Delta}_z(P)\} } \leq \eta.
\end{align*}
The last equality is due to $\sum_{z \in \scenarios_2} \beta^*_z = 1 - \sum_{z \in \scenarios_1} \beta^*_z$.
The proof is concluded since
\[opt(\ref{LP:EX-D}) \leq \sum_{e \in E}{u_e\sum_{0\leq t < T}{\sum_{z \in \scenarios}{\alpha_e^t(z)}}} \leq \eta \sum_{e \in E}{u_e\sum_{0\leq t < T}{\sum_{z \in \scenarios}{\bar{\alpha}_e^t(z)}}}\leq O(\eta k \log T) \cdot opt(\ref{LP:TR-D}).\] \qed
\end{proof}

\section{Open Problems}
\label{sec:conclusions}

In this work, we provided a first step towards modeling and solving robust flow over time problems.
We have shown that temporally repeated flows under the presence of uncertainty are no longer optimal.
We provided lower and upper bounds on the optimality gap.
Moreover, we have shown that the relation between delays, the time horizon and the longest path length has a strong impact on the complexity status.
We have shown that, for instances with $T$-bounded path length, an optimum temporally repeated solution can be computed in polynomial time.

Clearly, many interesting questions remain open.
We want to point out a few of these explicitly which may inspire follow-up work.

In our opinion, one of the biggest questions is clearly the complexity status of robust maximum flow over time if $\Gamma$ is bounded by a constant, even for $\Gamma = 1$.
Since the static counterpart is solvable in polynomial time for $\Gamma = 1$, it is interesting to see if the complexity status changes already due to the introduction of travel times.
The same question could be asked for temporally repeated solutions.
Although we were able to show that temporally repeated solutions can
be computed in polynomial time if the instance has $T$-bounded path
length, we were not able to provide insight for the case in which this
setting is not true.

A related question is the following.
Note that Proposition \ref{prop:temp-rep:LP-T-not} implies that, for $T$-bounded instances, there always exists an optimum temporally repeated solution which utilizes at most $2 |E|$ paths.
This is due to the number of constraints in the reformulated LP which we solve.
It would be interesting to understand if, for general instances or even for general solutions, there always exists an optimum solution which utilizes only a polynomial number of paths.
The LP formulations used in Section \ref{sec:bounds} contain an exponential number of variables and constraints.
Hence, it is unclear whether there always exists a basic feasible solution which is \emph{not} of exponential size.
If one could show that there are instances for which every optimum solution is of exponential size, one should also ask if the gap between polynomial size solutions and exponential size solutions can be bounded and if such solutions can also be computed in polynomial time.

The optimality gaps provided in Section \ref{sec:bounds} are not necessarily tight.
We have seen that for $k = \eta = 1$ in $T$-bounded instances, the lower bound from Proposition \ref{prop:temp-rep:log-gap} and the upper bound from Theorem \ref{thm:temp-rep:delta-arbitrary} coincide up to constant factors.
Apart from that, larger gaps remain.
It would be interesting to close these gaps, for example by constructing instances for $k$-coverable instances with $k > 1$ matching the upper bound from the Theorem.
This is especially interesting for instances with $\Delta \in \{0,\infty\}$ as discussed in Theorem \ref{thm:temp-rep:k-log-T-gap}.
Here, the gap between our lower and upper bounds is even bigger.
Moreover, it would be interesting to see if the dependency on $\eta$ in Theorem \ref{thm:temp-rep:delta-arbitrary} is necessary.
In the proof strategy we used, we were not able to remove it, although our lower bound instances always satisfied $\eta = 1$ and we were not able to construct gap instances with $\eta > 1$.

As mentioned before, the model studied here is a natural formulation for robust flow over time problems under uncertain travel times, using the well-established $\Gamma$-robustness model.
Due to the worst-case nature of the models introduced here, the resulting robust counterpart is quite restrictive and in general yields conservative solutions.
However, its study is interesting in order to understand robust flow over time problems.
Based upon the results obtained here, more advanced and potentially less conservative models could also be studied in the future.

\section*{Acknowledgments}

We thank the reviewers for their very careful reading of the manuscript and their valuable comments.
We thank the DFG for their support within Project B06 in CRC TRR 154.

\Urlmuskip=0mu plus 1mu

\end{document}